\documentclass[12pt]{article}
\pdfoutput=1
\usepackage{amssymb,amsmath,amsthm}
\usepackage{graphicx, xcolor, varwidth}
\usepackage[percent]{overpic}
\usepackage{mathtools}
\usepackage{setspace}
\usepackage{cite}
\usepackage[colorlinks=true, allcolors=darkblue, urlcolor=darkblue, linktocpage=true, linkcolor=darkblue, citecolor=purple]{hyperref}
\usepackage{float}
\usepackage[most]{tcolorbox}

\tcbset{
    frame code={}
    center title,
    left=0pt,
    right=0pt,
    top=0pt,
    bottom=0pt,
    colback=gray!70,
    colframe=white,
    width=\dimexpr\textwidth\relax,
    enlarge left by=0mm,
    boxsep=5pt,
    arc=0pt,outer arc=0pt,
    }

\colorlet{darkblue}{blue!70!black}
\colorlet{darkgreen}{green!70!black}

\newcommand{\nv}{\nu}

\newcommand{\mrm}{\mathrm}

\newcommand{\nn}{\nonumber}

\newcommand{\lb}{\left(}
\newcommand{\rb}{\right)}
\newcommand{\lcb}{\left\{}

\newcommand{\FF}{\mathcal{F}}

\newcommand{\rcb}{\right\}}
\newcommand{\lsb}{\left[}
\newcommand{\rsb}{\right]}

\newcommand\be{\begin{equation}}
\newcommand\ba{\begin{eqnarray}}
\newcommand\ee{\end{equation}}
\newcommand\ea{\end{eqnarray}}

\newcommand\padic{$p$-adic }

\numberwithin{equation}{section}

\newcommand{\arxivold}[1]
  {\href{http://arxiv.org/abs/#1}{#1}}

\newcommand{\arxiv}[1]
  {\href{http://arxiv.org/abs/#1}{arXiv:#1}}

\DeclareMathOperator{\sgn}{sgn}
\newcommand\Qp{{\mathbb{Q}_p}}
\newcommand\Zp{{\mathbb{Z}_p}}

\newtheorem{defn}{Definition}

\newtheorem{thm}{Theorem}
\newtheorem{lemma}{Lemma}
\newtheorem{remark}{Remark}

\newtheorem{sdefn}{Sketch definition}

\newtheorem{sthm}{Sketch theorem}
\newtheorem{slemma}{Sketch lemma}
\newtheorem{sremark}{Sketch remark}

\usepackage{cleveref}

\graphicspath{ {./images/} }

\begin{document}

\begin{titlepage}

\ \\
\vspace{-2.3cm}
\begin{center}

\begin{spacing}{2.3}
{\LARGE{Green's Functions for Vladimirov Derivatives and Tate's Thesis}}
\end{spacing}

\vspace{0.5cm}
An Huang,$^{1}$ Bogdan Stoica,$^{2}$ Shing-Tung Yau,$^{3,4}$ and Xiao Zhong$^{5}$

\vspace{5mm}

{\small

\textit{
$^1$Department of Mathematics, Brandeis University, Waltham, MA 02453, USA}\\

\vspace{2mm}

\textit{
$^2$Department of Physics and Astronomy, Northwestern University, Evanston IL 60208,~USA}\\

\vspace{2mm}

\textit{
$^3$Department of Mathematics, Harvard University, Cambridge MA 02138, USA}

\vspace{2mm}

\textit{
$^4$Center of Mathematical Sciences And Applications, Harvard University, Cambridge MA 02138, USA }\\

\vspace{2mm}

\textit{
$^5$Department of Pure Mathematics, University of Waterloo, Waterloo, Ontario, N2L 3G1, Canada}

\vspace{4mm}

{\tt anhuang@brandeis.edu, bstoica@northwestern.edu, yau@math.harvard.edu, xiao.zhong@uwaterloo.ca}

\vspace{0.3cm}
}

\end{center}

\begin{abstract}
\noindent Given a number field $K$ with a Hecke character $\chi$, for each place $\nu$ we study the free scalar field theory whose kinetic term is given by the regularized Vladimirov derivative associated to the local component of $\chi$. These theories appear in the study of $p$-adic string theory and $p$-adic AdS/CFT correspondence. We prove a formula for the regularized Vladimirov derivative in terms of the Fourier conjugate of the local component of $\chi$. We find that the Green's function is given by the local functional equation for Zeta integrals. Furthermore, considering all places $\nu$, the field theory two-point functions corresponding to the Green's functions satisfy an adelic product formula, which is equivalent to the global functional equation for Zeta integrals. In particular, this points out a role of Tate's thesis in adelic physics.
\end{abstract}

\vspace{3cm}

\begin{center}
\emph{In memory of Steven Gubser and John Tate.}
\end{center}

\end{titlepage}

\pdfbookmark[1]{\contentsname}{toc}
\setcounter{tocdepth}{2}
\tableofcontents
\onehalfspacing
\clearpage

\section{Introduction}

The field of $p$-adic numbers and its many generalizations have been of interest in physics since the works of Manin \cite{Manintheta}, Freund and Olson \cite{FreundOlson}, Freund and Witten \cite{FreundWitten}, and others, tracing back to the time period immediately after the first superstring revolution. \mbox{$p$-adic} quantum mechanics has been considered in \cite{vladimirovvolovich1, vladimirovvolovich2, vladimirovvolovich3, ruelleetal, zelenovpeq2, zelenovpathintegral}; this series of works led to the introduction of the so-called Vladimirov derivative, which is the natural notion of derivative on a field of numbers with the $p$-adic topology, and will be a central object of investigation in our present paper. Vladimirov derivatives have played an important role in the works of Zabrodin \cite{zabrodin,zabrodin2} computing the $p$-adic Veneziano amplitudes as path integrals on the Bruhat-Tits tree, and which can be thought of as a precursor of AdS/CFT. In recent years, motivated in part by the appearance of discrete features such as tensor networks and quantum error-correction in holography, there has been a resurgence of interest in $p$-adic physics \cite{ManMar,Gubser:2016guj, Heydeman:2016ldy, Gubser:2016htz, Gubser:2017vgc, Bhattacharyya:2017aly, Gubser:2017tsi, Dutta:2017bja, Gubser:2017qed, Bocardo-Gaspar:2017atv, Marcolli:2018ohd, Gubser:2018bpe, Gubser:2018yec, Qu:2018ned, Jepsen:2018dqp, Gubser:2018cha, Heydeman:2018qty, Hung:2018mcn, Jepsen:2018ajn, Gubser:2019uyf, Garcia-Compean:2019jvk}.

\textbf{Motivation from physics:} The interest in $p$-adic number systems in physics stems from attempts to understand the ultraviolet (high-energy) behavior of physical theories. Many physical quantities in gravity and quantum field theory (such as scattering amplitudes) naively suffer from divergences coming from the high-energy regime of the theory, which must be ``renormalized away'' via various methods. These divergences are a consequence of the fact that the high-energy regime in these theories is not well-defined, and they are expected to be removed, or to become tamer, once the theories are UV-completed. The situation in string theory is better, and it is believed that the extended nature of strings (and of other objects in the theory) renders the string theory ultraviolet finite, however how this happens is not precisely understood. Furthermore, string theory has regimes (such as the strong string coupling regime) where the conventional perturbative methods typically used to describe it cannot be applied.

The point of $p$-adic approaches to physics is that these ultraviolet difficulties are intimately connected to the topological structure of the real numbers, where two points can be brought arbitrarily close together in a continuous manner. Thus, a sensible attempt at dealing with renormalization more systematically is to replace the the field of real numbers with fields where the topology is better behaved (such as the $p$-adic fields), define physical theories on these objects, and then attempt to ``reconstruct'' the original theory on the reals \cite{Stoica:2018zmi}. Such an approach could also help guide understanding the ultraviolet completion of various theories. 

Perhaps surprisingly, introducing $p$-adic fields makes contact with another feature present in modern high-energy physics, that of nonlocality. Various notions of nonlocality exist in high-energy physics, for instance in quantum entanglement, or in the holographic dictionary of the AdS/CFT correspondence, where local bulk data has nonlocal representations on the boundary, and vice versa. Because of the negative sign in the exponent of the $p$-adic norm, rational points that are close in $p$-adic space will typically be far apart in real space, and the other way around. This provides a natural (and fundamental) realization of nonlocality in the context of $p$-adic physics.

\textbf{Physics and mathematics:} Historically, developments in $p$-adic physics have been closely related to developments in mathematics. In the present paper we will contribute to this connection, by establishing results between the reconstruction of propagators, in the sense of \cite{Stoica:2018zmi}, and functional equations for Zeta functions and $L$-functions.

Our construction is as follows. Given any number field $K$ with any choice of a Hecke character,\footnote{As is standard in the mathematics and physics literature, we will denote by $\chi$ both the Hecke characters and the local additive characters. The notation should be clear from the context, but as a general rule of thumb $\chi$ will be used mostly for the local additive characters in Sections \ref{sec2}, \ref{secnumberfield}, and \ref{secQpn}, and for Hecke characters in Section \ref{cftnf}.} we associate a field theory on each local completion of $K$, which is defined by a free action whose kinetic term is determined by the local component of the Hecke character. When $K=\mathbb{Q}$ with the character $x\to|x|$, it turns out that the field theory on the local completion is the dual CFT of the genus 0 $p$-adic open string worldsheet theory, as first worked out by Zabrodin \cite{zabrodin2}. When $K=\mathbb{Q}(i)$, at the Archimedean place, the theory on the local completion is the two dimensional free scalar field theory. Furthermore, the product of these local completion two point functions (in the sense of analytic continuation) is equal to 1, which is ensured by the global functional equation in Tate's thesis. Thus we find a role of Tate's thesis in adelic physics, with immediate generalizations anticipated -- e.g. to the quaternions over $\mathbb{Q}$.\footnote{There is an unfortunate clash of terminology between ``local'' as in local completion or component and the physicist's notion of locality. We should note however that for $K=\mathbb{Q}$ or $\mathbb{Q}(i)$ and character $x\to |x|$, the resulting Archimedean theories are local in the physicist's sense as well, because the Archimedean Green's functions are Green's functions for the usual (local) Laplacian with Dirac-delta source.} 

The key to relate Tate's thesis with adelic physics is to interpret the local Zeta integrals in terms of Green's functions for the so-called Vladimirov derivatives, which are linear operators from the Schwartz space to a bigger function space. This interpretation relies on the fact that there are two ways to express the Vladimirov derivative operator in general: either in terms of the Fourier conjugate of a multiplicative character, or in terms of an explicit integral formula. In Section \ref{secnumberfield} we will prove that these two definitions are equivalent in general.

There is one other point worth mentioning. In physics, one instance of the \mbox{AdS/CFT} (holographic) correspondence states that gravitational theories on asymptotically hyperbolic spaces are equivalent to conformal field theories on the boundary of these spaces. As originally noted by Zabrodin \cite{zabrodin2}, in the $p$-adic setting, theories on $\Qp$ have bulk duals on the Bruhat-Tits tree $T_p$. From this point of view, Zeta functions and $L$-functions can be understood as objects to which it makes sense to apply holography, and the work in the present paper represents a \emph{boundary} computation, whereas the previous work \cite{Huang:2019nog} was in the bulk. 

\textbf{The symmetries of $p$-adic scalar theories:} In this paper we will be mostly concerned with free scalar theories on $\mathbb{Q}_p$, $\mathbb{R}$, or a number field. These theories are given by a standard kinetic term action,
\be
S=\int_{K_{\nu}}\phi(x) \partial^{\chi_{\nv}}\phi(x)dx,
\ee
defined in relation to a Vladimirov derivative $\partial^{\chi_{\nv}}$. As will be explained in detail below, when $K_\nu=\mathbb{Q}_p$ and $\chi_{\nv}=|\cdot|^s_p$ (here $|\cdot|_p$ is the $p$-adic norm and $s\in\mathbb{C}$) this action becomes
\be
\label{eqac12}
S = \int_{\mathbb{Q}_p} \phi(x) \frac{\phi(x)-\phi(y)}{|x-y|^{s+1}_p} dxdy.
\ee
When $s=1$, this action is invariant under the action of $\mathrm{PGL\lb2,\mathbb{Q}_p\rb}$, and so can be thought of as a $p$-adic analogue of a conformal field theory (as originally employed by Zabrodin \cite{zabrodin2}, and explained in \cite{Heydeman:2016ldy}). When $s\neq 1$, the action \eqref{eqac12} remains invariant under scaling, but may no longer be invariant under more general transformations. Note that in the Archimedean setting conformal and scaling symmetry are closely related, but this relation appears weaker in $p$-adic settings. When $\chi_\nu$ is replaced by a more general character, the symmetries of the action \eqref{eqac12} become more complicated, and can depend for example on the behavior of the Hilbert symbol under scaling. A preliminary study of $p$-adic conformal field theory was performed by Melzer \cite{Melzer}, however in the present paper we will be mostly agnostic about conformal symmetry and the symmetries of the action, since they will not affect our analysis. 

We should also emphasize that the theories considered in this paper are in general (for arbitrary character) non-local, since the Vladimirov derivative has a non-local presentation. However, theories that are non-local at the finite places can be local at the Archimedean place; the free particle in quantum mechanics, particle in a box, as well as $p$-adic bosonic strings (i.e. the theories discussed in this paper when $K_\nu=\mathbb{Q}_p$ and $s=1$, and their Bruhat-Tits tree duals) are examples of this phenomenon. Furthermore, it is currently not well understood what forms restrictions such as unitarity, causality, and so on take in non-Archimedean theories, or how they are related between the non-Archimedean and Archimedean sides.

A connection between string four-point amplitudes, Tate's thesis, and the Riemann Zeta function was reported in \cite{AdelicNpoint}.\footnote{We thank the anonymous referee for pointing out this paper.} Our work differs from \cite{AdelicNpoint} in that we are working with the Green's functions, and not with the four-point amplitudes, and furthermore our analysis applies not only to the Zeta functions, but also to a Hecke $L$-function of a general number field. While it may be possible to formally write down amplitudes related to such $L$-functions, currently it is not clear to which theories they should belong. Nonetheless, it would be worthwhile to systematically study the relations between Green's functions and four- and $N$-point amplitudes, from the viewpoint of Tate's thesis.

\subsection{Summary of results and outline}

We now summarize our results. The main observation of this paper is that the product over places of the Green's functions for derivatives,
\be
\label{prod111}
\prod_{\nu} G_{(\nu)}(x,y)=1,
\ee
for $\mathbb{Q}$ and other number fields, is equivalent to a global functional equation. When $\pi=|\cdot|^s$ the functional equation is for the Riemann Zeta function, while for more general quasi-characters on number fields it will be an $L$-function functional equation. The equivalence between the product \eqref{prod111} and functional equations follows directly from Tate's thesis.

As we will show, the Green's function for the Vladimirov derivative is unique up to an additive constant. Equation \eqref{prod111} corresponds to picking this constant so that the Green's function is equal to a two-point function in field theory. Therefore, the adelic product formulas can be more precisely thought of as applying to two-point functions in field theory. Throughout the rest of the paper, we will use the terms ``Green's functions'' and ``two-point functions'' interchangeably when discussing product formulas, with the understanding that the constant in the Green's functions has been chosen appropriately.

At a finite place, the Vladimirov derivative $D_\pi$ acting on functions $\psi:F\to\mathbb{C}$ has two representations, 
\be
D^{(1)}_\pi \psi(x) = \int_{F} \chi(-kx)\pi(k) \int_{F} \chi(kx') \psi(x') dx' dk,
\ee
where $\chi$ is an additive character, and
\be
D^{(2)}_\pi \psi(x) =\Gamma(|\cdot|\pi) \int_{F} \frac{\psi(x) - \psi(x')}{\pi(x-x')|x-x'|} dx.
\ee
where the $|\cdot|: F \longrightarrow \mathbb{R}^{\geq 0}$ is the standard ultrametric non-Archimedean absolute value of this finite place of $F$. We prove that the two representations are equal when acting on functions in the Schwartz space, i.e. $D_\pi^{(1)}=D_\pi^{(2)}$, for non-Archimedean completions of number fields. The identity $D_\pi^{(1)}=D_\pi^{(2)}$ can be thought of as bridging physics and Tate's thesis, as representation $D_\pi^{(2)}$ is natural in the context of Tate's thesis, whereas representation $D_\pi^{(1)}$ is common in physics.

The outline of the paper is as follows. In Section \ref{sec2} we lay out the skeleton of our argument, and we exemplify it for the case of $F=\Qp$. We present each step as a \emph{sketch}, for arbitrary $F$, and then we make these steps precise for $F=\Qp$. These results for $\Qp$ are not new, as they have appeared before in various parts of the literature (see e.g. \cite{vvzbook} for an early reference, or \cite{zunigabook} for a rigorous treatment), and in any case they follow from Tate's thesis. Rather, we intend this section as a review of how the machinery works, for readers who may not be familiar with Tate's thesis, and as a bare-bones presentation of the structure that generalizes to more involved cases. In Section \ref{secnumberfield} we extend our analysis to number fields, by showing that the equivalence between the two presentations of the Vladimirov derivative holds generally. As far as we know, the computations in this section are new. Furthermore, the equivalence between the Euler product of Green's functions and the functional equation again holds, following from Tate's thesis. In Section \ref{cftnf} we comment on applications of the machinery in this paper to $p$-adic field theory. In Section \ref{secQpn} we consider the case of multiplicative characters on $\mathbb{Q}_p^n$ coming from the maximum norm.

\section{Green's functions for Vladimirov derivatives on~$\Qp$}
\label{sec2}

\subsection{The general recipe}
\label{sec21}

We now present our argument. It works in much the same way on any arbitrary $F$ that is a local completion of a number field, however the specifics (such as for which values of the parameters the integrals converge, and so on) will differ from one case to another. In this section we will specify $F=\Qp$. Nonetheless, the reader should keep in mind that the argument is written so that it can be extended to more general situations. We will discuss the case when $F$ comes from a number field in Section \ref{secnumberfield}, and we will touch on $F=\mathbb{Q}_p^n$ in Section \ref{secQpn}~below. 

We will first sketch the main arguments below ignoring certain details, including some convergence issues. These details will be recovered later in Section \ref{sec232} of the paper for $\mathbb{Q}_p$ by direct computation, and in general for a local completion of a number field by Tate's thesis, discussed in Sections \ref{secnumberfield} and \ref{cftnf}. We are interested in functions (which could for instance be the fields in a conformal field theory)
\be
\phi: F \to \mathbb{C}.
\ee  

The kinetic term in the Lagrangian is a key object in any physical theory of the fields $\phi$ with a Lagrangian description. Kinetic terms on non-Archimedean fields can be constructed by introducing Vladimirov derivatives, which are the Fourier transforms of quasi- (multiplicative and continuous) characters. More precisely, a quasi-character $\pi:F\to\mathbb{C}$, $\pi(x)\pi(y)=\pi(xy)$, imparts on $F$ a certain structure, which can be exploited in order to build a physical theory, in the form of the Gamma function $\Gamma\lb\pi\rb$, the Vladimirov derivative $D_\pi$, and the Green's function $G_\pi(x,y)$ for the Vladimirov derivative. The choice of quasi-character can thus be thought of as specifying the kinetic term for the theory, and the entire structure that we will be using here (Gamma function, Vladimirov derivatives, Green's functions), i.e. the physics, essentially arises from the representation theory on the multiplicative group.

\begin{defn}
Additive $\chi$ and multiplicative $\pi$ characters of $F$ are functions
\be
\chi:F \to \mathbb{C}^\times, \quad \pi:F^\times \to \mathbb{C}^\times,
\ee
such that for $x,y\in F^\times$ and $w,t\in F$ we have
\be
\pi(xy)=\pi(x)\pi(y),\quad \chi(w+t) = \chi(w)\chi(t).
\ee
For $F=\Qp$ we will choose additive characters
\be
\chi(x)\coloneqq e^{2\pi i \{x\}},
\ee
where $\{x\}$ is the fractional part of $x$, i.e. if $ x = \sum_{n} c_n p^n $ then $\{x\} = \sum_{n < 0} c_n p^n$, and multiplicative characters given by (of course, other choices of characters are possible also)
\be
\pi_s\coloneqq |x|_p^s,\quad \pi_{s,\tau}\coloneqq|x|_p^s\sgn_\tau x,
\ee
with $s\in\mathbb{C}$,  $\tau\in\mathbb{Q}_p^\times$ and $sgn_\tau x = (\tau ,x)$ is the Hilbert symbol, so that $\pi_s=\pi_{s,1}$. Throughout the paper (and especially in Section \ref{secnumberfield}), we will be interested in characters that are continuous on the multiplicative group, i.e. quasi-characters, and furthermore we will assume $\Re\lb s\rb>0$.
\end{defn}

\begin{remark}
We will ignore convergence issues throughout this section (for instance the Sketch theorem \ref{Green Thm} and Sketch lemma \ref{slemma2} below). It will turn out that things can be made rigorous either by direct computation in the $\mathbb{Q}_p$ case (see Section \ref{sec232}), or in general for a local completion of a number field by Tate's thesis.
\end{remark}

\begin{sdefn}
The Fourier transform of a function $\phi$ in the Schwartz space of compactly supported locally constant functions on $F$ (integration is with respect to the additive Haar measure on $F$) is
\be
\FF\phi(k) \coloneqq \int_F \phi(x) \chi(kx)dx.
\ee
\end{sdefn}

\begin{slemma}
The inverse of the Fourier transform is
\be
\FF^{-1}\phi(k) = \FF\phi(-k).
\ee
\end{slemma}

\begin{sdefn}
\label{defvladi}
The Vladimirov derivative $D_\pi$ associated to character $\pi$ is
\be
\label{VladiDpi}
D_\pi \phi \coloneqq \FF^{-1} \pi \FF\phi.
\ee
\end{sdefn}

\begin{remark}
Historically there is another explicit formula for the Vladimirov derivative for $\mathbb{Q}_p$ and an unramified character, which is used for example in Zabrodin's work on $p$-adic string amplitudes. We will prove later that these two definitions are equivalent in general.
\end{remark}

\begin{sremark}
The crucial reason why Definition \ref{defvladi} makes sense is that for two quasi-characters $\pi_{1,2}$ we have
\be
D_{\pi_1}D_{\pi_2} = D_{\pi_1\pi_2},
\ee
so the derivatives compose as expected.
\end{sremark}

Let's now introduce the Green's function $G(x,y)$ for $D_\pi$.
\begin{sdefn}
For $x,y\in F$, let $G(x,y)$ be the function
\be
G_\pi(x,y)\coloneqq \FF \pi^{-1}(x-y).
\ee
\end{sdefn}

\begin{sthm}
\label{Green Thm}
Function $G(x,y)$ is the Green's function for the Vladimirov derivative $D_\pi$.
\end{sthm}

\begin{proof}
Fourier transforming and applying $\pi$, we have
\be
\pi \FF^{-1} G = 1,
\ee
which Fourier transforming again gives
\be 
\FF \pi \FF^{-1} G = \delta,
\ee
so that
\be
D_\pi G(x,y) = \delta(x-y),
\ee
which shows that $G$ is indeed the Green's function.
\end{proof}

The following lemma introduces the Gamma function, and can be used to characterize the Vladimirov Green's functions in more detail.

\begin{slemma}
\label{slemma2}
The Fourier transform of a quasi-character $\pi$ is another quasi-character, given by
\be
\label{GammaFTgeneral}
\FF \pi = \Gamma\lb \Delta \pi \rb \Delta^{-1} \pi^{-1},
\ee
where prefactor $\Gamma\lb \Delta \pi \rb$ is called the Gamma function, and $\Delta$ is the Haar modulus. For $F=\Qp$, we have $\Delta(x) = |x|_p$.
\end{slemma}
\begin{proof}
By the definition of the Fourier transform, we have
\be
\FF\pi (k) = \int_F \chi\lb kx \rb \pi(x) dx.
\ee
Consider the change of variables $x'\coloneqq kx$. This gives
\be
\label{eq114}
\FF \pi\lb k \rb = \lb \int_F \chi\lb x' \rb \pi\lb x' \rb dx' \rb \pi\lb k^{-1} \rb \Delta^{-1}(k),
\ee
as desired.
\end{proof}

\begin{sdefn} According to Eqs. \eqref{GammaFTgeneral}, \eqref{eq114}, the Gamma function for character $\pi$ is defined as
\be
\label{eqgammadef}
\Gamma\lb \pi \rb \coloneqq \int_F \chi(x) \pi\Delta^{-1}(x)dx.
\ee
\end{sdefn}

\begin{slemma}
The Gamma function obeys the functional equation
\be
\Gamma\lb \Delta \pi \rb \Gamma\lb \pi^{-1} \rb = \pi(-1).
\ee
For $F=\Qp$ and characters $\pi_{s,\tau}$, this functional equation is just
\be
\label{Gammastpart}
\Gamma\lb \pi_{s+1,\tau} \rb \Gamma\lb \pi_{-s,\tau} \rb = \sgn_\tau(-1).
\ee
\end{slemma}
\begin{proof}
Applying $\FF^{-1}$ to Eq. \eqref{GammaFTgeneral} we have
\ba
\pi(k) &=& \Gamma\lb \Delta \pi \rb \FF^{-1}\lb \Delta^{-1} \pi^{-1}\rb(k) \\
&=& \Gamma\lb \Delta \pi \rb \Gamma\lb \pi^{-1} \rb\pi(-k),
\ea
so  that setting $k=-1$ and using $\pi(1)=1$ we obtain
\be
\label{gammaFUN}
\Gamma\lb \Delta \pi \rb \Gamma\lb \pi^{-1} \rb = \pi(-1),
\ee
from which Eq. \eqref{Gammastpart} also follows.
\end{proof}

\begin{slemma} Using Sketch Lemma \ref{slemma2}, the Green's function is
\be
\label{sketchG}
G_\pi\lb x-y \rb = \Gamma\lb \Delta \pi^{-1} \rb \Delta^{-1} \pi (x-y).
\ee
\end{slemma}

\begin{remark}
It is especially important to note that Sketch Eq. \eqref{sketchG} above assumes the analyticity of the Gamma function $\Gamma\lb\Delta\pi^{-1}\rb$. If $\Gamma\lb\Delta\pi^{-1}\rb$ has a pole, then the expression \eqref{sketchG} for the Green's function will not be correct; this happens for instance, when $F=\Qp$ and $\pi(x)=|x|^s$, at $s=1$, as we will review in Section \ref{sec232} below. In such cases, however, the Green's function can still be obtained from Eq. \eqref{sketchG} by taking the~limit.
\end{remark}

\begin{remark}
There exists an integral presentation of the Vladimirov derivative~\eqref{VladiDpi}, given~by
\be
\label{Dinteg}
D_\pi \psi\lb x\rb = \Gamma\lb \Delta \pi \rb \int_F \lsb \psi(x') - \psi(x) \rsb \Delta^{-1}\lb x'-x \rb \pi^{-1}\lb x'-x \rb dx'.
\ee
\end{remark}
\noindent This formula is well-known in the case $F=\Qp$, when $\Delta=|\cdot|$. In Section \ref{secnumberfield} we will prove that this presentation remains valid on number fields.

We will not prove formula \eqref{Dinteg} here, but it is immediate to justify it, since starting with Eq. \eqref{VladiDpi} and using Eq. \eqref{GammaFTgeneral} for the Fourier transform of the multiplicative character, we have
\ba
\label{eeq225}
D_\pi \psi(x) &=& \int_F \pi(k) \chi\lsb k\lb x'-x \rb \rsb \psi(x')dkdx'\\
&=&\Gamma\lb \Delta \pi \rb \int_F \psi(x') \lb \Delta^{-1} \pi^{-1}\rb\lb x'-x \rb dx'.
\label{eeq226}
\ea
This recovers the first term in Eq. \eqref{Dinteg}. The derivation in Eqs. \eqref{eeq225} -- \eqref{eeq226} is not rigorous, because it is not in the sense of distributions, and we are being cavalier about convergence. Taking care of these issues properly also produces the second (contact) term in Eq. \eqref{Dinteg}; see Section \ref{secnumberfield} for the details.

\begin{slemma}
\label{slemma5}
Suppose the quasi-characters and Haar modulus obey the product formulas
\be
\prod_v \pi^{(v)}=1, \quad \prod_v \Delta^{(v)} = 1,
\ee
when multiplied across all places $\{F_v\}$ (Archimedean and non-Archimedean). Then the Green's functions obey the product formula
\be
\label{prod223}
\prod_v G^{(v)}_{\pi^{(v)}}(x-y) = \prod_v \Gamma^{(v)} \lb \Delta^{(v)} \lb\pi^{(v)}\rb^{-1} \rb.
\ee
\end{slemma}

\begin{proof}
Immediate from Eq. \eqref{sketchG}.
\end{proof}

We now discuss the product in Eq. \eqref{prod223}, in the case $F=\Qp$.

\begin{remark}
It is a standard result that integrating Eq. \eqref{eqgammadef}, the Gelfand-Graev Gamma function for $\pi_s$ is
\be
\Gamma^{(p)}\lb \pi^{(p)}_s \rb = \frac{1-p^{s-1}}{1-p^{-s}}.
\ee
Similarly, at the Archimedean place we have the Gelfand-Graev Gamma function
\be
\Gamma^{(\infty)}\lb \pi^{(\infty)}_s \rb = \frac{2\cos\lb\frac{\pi s}{2}\rb}{(2\pi)^s}\Gamma_\mrm{Euler}(s).
\ee
One can compute that this Archimedean factor comes out in exactly the same way in the Green's function at the Archimedean place: namely, in this case, compute the Green's function for the pseudodifferential operator at the real place, the resulting Green's function has exactly this Archimedean Gamma factor in front of a quasi-character. For details of this, one can see Tate's thesis, where the same computation is done, but is interpreted in a different way.
\end{remark}

Tate's thesis also shows that the product of the Gamma functions for the local components $\pi_s$ being equal to $1$ is equivalent to the functional equation for the Riemann Zeta function, which we encapsulate in Theorem \ref{thmone} below. In light of Sketch Lemma \ref{slemma5} above, we translate this into the following: Tate's thesis states that the Archimedean propagator for a free bosonic field theory being reconstructible from the $p$-adic propagators via a product formula is equivalent to the functional equation for the Riemann Zeta function.

\begin{thm} The product
\label{thmone}
\be 
\prod_v \Gamma^{(v)}(\pi^{(v)}_s) = 1
\ee
is equivalent to the Riemann Zeta functional equation
\be
\label{zetafnct}
\zeta\lb 1-s \rb = 2^{1-s}\pi^{-s} \cos\lb \frac{\pi s}{2} \rb \Gamma_\mrm{Euler}\lb s \rb \zeta\lb s\rb.
\ee
\end{thm}
\begin{proof} This is immediate by direct computation, as we have
\ba
\prod_v \Gamma^{(v)}(\pi^{(v)}_s) &=& \Gamma^{(\infty)}(\pi^{(\infty)}_s) \prod_p \Gamma^{(p)}(\pi^{(p)}_s) \\
&=& \Gamma^{(\infty)}(\pi^{(\infty)}_s) \prod_p \frac{1-p^{s-1}}{1-p^{-s}} \\
&=& \frac{2\cos\lb\frac{\pi s}{2}\rb}{(2\pi)^s}\Gamma_\mrm{Euler}(s) \frac{\zeta(s)}{\zeta(1-s)}.
\ea
The last step above requires analytic continuation, which we will not justify here, as it is justified in Tate's thesis.
\end{proof}

\begin{thm} The product
\be
\prod_v \Gamma^{(v)}(\pi^{(v)}_{s,\tau}) = 1
\ee
is equivalent to the functional equation for quadratic character Dirichlet $L$-functions for the Hecke character $\chi$, of which the $\pi_{s,\tau}$'s are the local components.
\begin{proof}
Just as in the Riemann Zeta case, expanding the Euler product gives the local factors of the $L$-function. The details are given in Tate's thesis.
\end{proof}
\end{thm}

\subsection{A presentation of the Vladimirov derivative on $\Qp$}
\label{sec23}

We now present the explicit formulas for the Vladimirov derivative operators on $\Qp$. In this section, we denote functions from Schwartz space as $\psi(x)$.
\begin{defn}
	The derivative operator $D_1$ on Schwartz space of $\Qp$ is defined as 
	\be
	D^s_1 \psi(x) = \int_{\Qp} e^{-i2\pi\{kx\}} |k|^s\int_{\Qp} e^{i2\pi\{kx'\}} \psi(x') dx'dk.
	\ee
\end{defn}

\begin{thm}
	The Fourier transform of a quasi-character is another quasi-character. That is,
	\be \label{gam}
	\frac{1- p^{s-1}}{1 - p^{-s}}|k|^{-s} =\int_{\Qp} e^{i2\pi\{kx\}}|x|^{s-1}dx.
	\ee
\end{thm}
\begin{proof}[Sketch proof]
	This result is well-known, see e.g. \cite{zunigabook}. By direct computation, denoting $v_p(k) = n$,
	\ba
	\int_{\Qp} e^{i2\pi\{kx\}}|x|^{s-1}dx  &=& \frac{p^{sn}(1- p^{-1})}{1- p^{-s}} - p^{s(1 +n)-1}\\
	&=&|k|^{-s}\frac{1- p^{s-1}}{1 - p^{-s}}.
	\ea
\end{proof}

\begin{defn}
	The derivative operator $D_2$ on the Schwartz space is defined as 
	\be
	D_2 \psi(x) \coloneqq \Gamma(\pi_{s+1})\int_{\Qp}\frac{\psi(x') - \psi(x)}{|x'-x|^{s+1}} dx.
	\ee
	Note that what we really mean is to integrate over  $\Qp - {x}$. 
\end{defn}

\begin{thm}
	Derivative operators $D_1$ and $D_2$ are equivalent in the sense that for all $\psi(x)$ in the test function space, we have
	\be
	\label{D1D2areeq}
	D_1 \psi(x) = D_2 \psi(x).
	\ee
\end{thm}
\begin{remark}
Eq. \eqref{D1D2areeq} is well-known in the literature, see for instance \cite{zunigabook}. We will extend this result to Vladimirov derivatives for quasi-characters on non-Archimedean completions of number fields, in Section \ref{secnumberfield}.
\end{remark}

\subsection{Example: The Vladimirov derivative for $\pi_s$ on $F=\Qp$}
\label{sec232}

In this section, we demonstrate how to prove the local Green's function formula by direct computation, for the case of $\mathbb{Q}_p$ with an unramified local quasi-character. In Section \ref{cftnf} we will prove a more general version of this computation, for number fields; because of this, our presentation here will be brief.

Throughout this section we will assume that the parameter $s\in\mathbb{C}$ labeling the quasi-characters and Vladimirov derivatives obeys $\Re(s)>0$. We normalize the volume of $\Zp$ as 
\be
\mrm{vol}\lb \mathbb{Z}_p \rb = 1,
\ee
so that, by the properties of the Haar measure, the volume of the $p$-adic circle
\be
\label{volS}
S_{i}\coloneqq \lcb x\in\mathbb{Q}_p\ |\ v_p(x)=i \rcb,
\ee
with $i\in\mathbb{Z}$, is given by
\be
\mrm{vol} \lb S_i \rb =  p^{-i}-p^{-i-1}.
\ee

\begin{lemma} 
\label{lemmasadj}
For $\Re(s)>0$, the Vladimirov derivative is a self-adjoint operator from the Schwartz space of compactly supported locally constant functions, to a bigger function space inside the space of continuous functions on $\mathbb{Q}_p$, in the sense that \eqref{SA} holds.
\end{lemma}
\begin{proof}
For arbitrary compactly supported locally constant functions $\psi_{1,2}$, consider the bracket

\ba\label{SA}
&&\langle \psi_1, D^s \psi_2 \rangle = \Gamma(\pi_{s+1}) \int_{\mathbb{Q}_p} \psi_1^*(x) \frac{\psi_2(z)-\psi_2(x)}{|z-x|^{s+1}} dzdx\\
&=& -\frac{1}{2} \Gamma(\pi_{s+1}) \int_{\mathbb{Q}_p} \frac{\lb \psi_1^*(z)- \psi_1^*(x) \rb\lb \psi_2(z)-\psi_2(x) \rb}{|z-x|^{s+1}}dzdx\\
&=& \langle \psi_2,D^s \psi_1 \rangle^*,
\ea
where we exchanged the order of integration by Fubini's theorem.
\end{proof}

\begin{remark}
Note that "self-adjoint" here only means formally that
$\langle \psi_2,D^s \psi_1 \rangle = \langle \psi_1, D^s \psi_2 \rangle^*$. In general the Vladimirov derivative operator doesn't preserve Schwartz space.
\end{remark}

\begin{defn}
For $x,y\in\mathbb{Q}_p$ and $s\in\mathbb{C}$ we introduce a function $G(x,y)$ as
\be
\label{Gishere}
G\lb x,y \rb\coloneqq \begin{cases}
c_{s,p} |x-y|^{s-1} \hspace{0.1cm}\quad \mrm{if} \quad s-1\neq 2\pi i k/\ln p, \quad k\in\mathbb{Z} \\
c_{s,p} \log_p |x-y| \hspace{0.22cm} \mrm{if} \quad s-1 =  2\pi i k/\ln p, \quad k\in\mathbb{Z}
\end{cases}, 
\ee
where constant $c_{s,p}$ is given by
\be
c_{s,p} = \begin{cases}
\Gamma\lb \pi_{1-s} \rb \hspace{1.05cm} \mrm{if}\quad  s-1 \neq 2\pi i k/\ln p, \quad k\in\mathbb{Z}\\
-\lb 1-p^{-1}\rb \quad \mrm{if}\quad  s-1 = 2\pi i k/\ln p, \quad k\in\mathbb{Z}
\end{cases}.
\ee
\end{defn}

\begin{thm} 
\label{hereisthm1}
For $x\neq y$ and $\Re(s)>0$ we have that 
\be
D^s_x G(x,y) = 0.
\ee
\end{thm}
\begin{proof} The proof proceeds by direct computation. In the case $s-1\neq 2\pi i k/\ln p$ we need to evaluate the integral
\be
I_n(s,v) \coloneqq \int_{\mathbb{Q}_p} \frac{|z-y|^{s-1} - |x-y|^{s-1}}{|z-x|^{s+1}} dz = \int_{\mathbb{Q}_p} \frac{|u|^{s-1} - |v|^{s-1}}{|u-v|^{s+1}} du,
\ee
where $u\coloneqq z-y$ and $v\coloneqq x-y$. Assuming $v\neq 0$, we have

\ba 
\label{eq416}
   &&I_n(s,v) = \int_{\mathbb{Q}_p} \frac{|u|^{s-1} - |v|^{s-1}}{|u-v|^{s+1}} du \\
&=& \int_{|u|<|v|} \frac{|u|^{s-1} - |v|^{s-1}}{|v|^{s+1}} du + \int_{|u|>|v|} \frac{|u|^{s-1} - |v|^{s-1}}{|u|^{s+1}} du \nn\\
&=& \sum_{k=\mrm{val}(v)+1}^\infty \frac{p^{-k(s-1)} - |v|^{s-1}}{|v|^{s+1}} \mrm{vol}\lb S_{k} \rb+ \sum_{k=-\infty}^{\mrm{val}(v)-1} \frac{p^{-k(s-1)} - |v|^{s-1}}{p^{-k(s+1)}} \mrm{vol}\lb S_{k} \rb \nn\\
&=& 0,\nn
\label{eq418}
\ea
as desired. Notice that the case $|u| = |v|$ has vanishing integral since $u =v$ has measure zero, and when $|u|=|v|$ but $u \neq v$ the integrand is zero. The sums converge when $\Re\lb s\rb>0$.

In the case $s-1= 2\pi i k/\ln p $ we need to evaluate the integral (when $v\neq0$)

\ba
&&J_n(v) \coloneqq \int_{\mathbb{Q}_p} \frac{\log_p |z-y| - \log_p |x-y|}{|z-x|^{s+1}} dz\\ &=& \int_{\mathbb{Q}_p} \frac{\log_p |u| - \log_p |v|}{|u-v|^{s+1}} du \nn\\
&=& \int_{|u|<|v|} \frac{\log_p |u| - \log_p |v|}{|v|^{s+1}} du + \int_{|u|>|v|} \frac{\log_p |u| - \log_p |v|}{|u|^{s+1}} du \nn\\
&=& \sum_{k=\mrm{val}(v)+1}^\infty \frac{-k - \log_p |v|}{|v|^{s+1}} \mrm{vol}\lb S_{k} \rb+ \sum_{k=-\infty}^{\mrm{val}(v)-1} \frac{-k - \log_p |v|}{p^{-k(s+1)}} \mrm{vol}\lb S_{k} \rb \nn \\
&=& 0.
\ea
For exactly the same reason, the integration with $|u| = |v|$ has been excluded. This completes the proof.
\end{proof}

\begin{remark}
Theorem \ref{hereisthm1} shows that $D_x^s G(x,y)$ applied to a test function can only depend on the value of the test function at $x=y$.
\end{remark}

\begin{thm}
When $\Re(s)>0$, $G(x,y)$ is the Green's function for the Vladimirov derivative $D_x^s$ on $\mathbb{Q}_p$, that is
\be
\label{mainresult1}
D_x^s G\lb x,y \rb = \delta \lb x-y \rb,
\ee
acting on the Schwartz space of compactly supported locally constant functions.
\end{thm}

We will prove this theorem in several steps (lemmas).

\begin{lemma}
Let $\chi_r$ be the characteristic function of the ball
\be
B_r \coloneqq \lcb x\in \mathbb{Q}_p\ |\ |x| \leq p^r \rcb.
\ee
$D^s_xG(x,y)$ acts as the delta function on the test function $\chi_r$.
\end{lemma}

\begin{proof}
We use direct computation of the bracket. By self-adjointness of the derivative we have
\be
\langle D^s G, \chi_r \rangle = \langle G, D^s\chi_r \rangle.
\ee
\textbf{Case 1.} Let's first consider the case $s-1 \neq 2\pi i k/\ln p$. We have
\ba
&&\frac{1}{c_{s,p}\Gamma(\pi_{s+1})}\langle G, D^s\chi_r \rangle = \int_{\mathbb{Q}_p} |x-y|^{s-1} \frac{\chi_{r}\lb z \rb - \chi_{r}\lb x \rb}{|z-x|^{s+1}}dzdx \\
&=& \lb \int_{x\notin B_r} \frac{|x-y|^{s-1}}{|x|^{s+1}}dx\rb\lb \int_{z\in B_r} dz \rb\\
& & - \lb\int_{x\in B_r} |x-y|^{s-1}dx\rb \lb \int_{z\notin B_r} |z|^{-s-1} dz \rb.\nn
\label{eq323}
\ea
We now compute the two $z$ integrals, which converge when $\Re\lb s\rb>0$,
\ba
\label{eqq324}
\int_{z\in B_r} dz &=& \sum_{j=-\infty}^r \mrm{vol} \lb S_{-j} \rb = p^{r},\\
\int_{z\notin B_r} |z|^{-s-1} dz &=& \sum_{j={r+1}}^\infty p^{-(s+1)j} \mrm{vol} \lb S_{-j} \rb \nn\\
&=&\lb 1-p^{-1} \rb \frac{p^{-sr}}{p^{s}-1}. \nn
\ea
Now we need to compute the $x$ integrals. There are two subcases: (i) $y\in B_r$ and (ii) $y\notin B_r$.

\textbf{(i)} Let's start with $y\in B_r$. We have
\ba
\int_{x\notin B_r} \frac{|x-y|^{s-1}}{|x|^{s+1}}dx &=& \int_{x\notin B_r} |x|^{-2}dx\\ = \sum_{j={r+1}}^\infty p^{-2j} \mrm{vol} \lb S_{-j} \rb &=& p^{-(r+1)}.\nn
\ea
For the $x\in B_r$ integral, note that $x'\coloneqq x-y\in B_r$, so that we can shift the integral to compute
\ba
\int_{x\in B_r} |x-y|^{s-1}dx &=& \int_{x'\in B_r} |x'|^{s-1}dx' \\
&=& \sum_{i=-\infty}^r p^{i\lb s-1 \rb} \mrm{vol} \lb S_{-i} \rb\\
&=& \lb 1-p^{-1} \rb \frac{p^{(r+1) s}}{p^{s}-1}.
\ea
Putting everything together, we obtain, when $y\in B_r$,
\be
\frac{1}{c_{s,p}\Gamma(\pi_{s+1})}\langle G, D^s\chi_r \rangle = \frac{\left(1-p^{s-1}\right)
   \left(p^{-1}-p^{s}\right)}{\left(1-p^{s}\right)^2},
\ee
which implies that $\langle G, D^s\chi_r \rangle=1$, as desired.

\textbf{(ii)} Let's now consider the subcase $y\notin B_r$, and we once again compute the $x$ integrals. We have
\ba
\label{eq328}
\int_{x\in B_r} |x-y|^{s-1}dx &=& |y|^{s-1} \int_{x\in B_r}dx\\ &=& |y|^{s-1} (1-p^{-1}) \sum_{i=-\infty}^r p^{i} = |y|^{s-1} p^{ r}.
\ea
The last remaining integral is marginally more involved. Let $|y|\eqqcolon p^{r'}$. We have
\ba
&&\int_{x\notin B_r} \frac{|x-y|^{s-1}}{|x|^{s+1}}dx = |y|^{s-1} \int_{\substack{x\notin B_r\\|x|<|y|}} |x|^{-s-1}dx \\
& &  |y|^{-s-1} \int_{|x|=|y|} |x-y|^{s-1}dx + \int_{|x|>|y|} |x|^{-2}dx.\nn
\ea
We now compute these three integrals, as
\ba
\label{eq330}
\int_{\substack{x\notin B_r\\|x|<|y|}} |x|^{-s-1}dx
&=& \sum_{i=r+1}^{r'-1} p^{-i(s+1)} \mrm{vol} \lb S_{-i} \rb\\
&=& \frac{\left(1-p^{-1}\right) \left(p^{-s(r+1)}-p^{-sr'}\right)}{1-p^{-s}},\nn\\
\int_{|x|>|y|} |x|^{-2n}dx &=& \sum_{i=r'+1}^\infty p^{-2ni} \mrm{vol} \lb S_{-i} \rb = p^{-1 (r'+1)},\nn
\ea
and finally for the middle term in Eq. \eqref{eq330}, using Eq. \eqref{eq3p9} in Lemma \ref{lemma2} below, we obtain
\be
\label{eq331}
|y|^{-s-1}  \int_{|x|=|y|} |x-y|^{s-1}dx = -\frac{p^{-1 r'} \left(p+\left(p-2\right)
   p^{s+1}\right)}{p^{2 }-p^{s+2n}}.
\ee
Putting together Eqs. \eqref{eq330} and \eqref{eq331}, we derive
\be
\label{eq332}
\int_{\substack{x\notin B_r\\|x|<|y|}} |x|^{-s-1}dx = -\frac{\left(1-p^{-1}\right) p^{-sr+r'(s-1)}}{1-p^{s}},
\ee
so that from Eqs. \eqref{eq323} \eqref{eqq324}, \eqref{eq328} and \eqref{eq332} we conclude (when $y\notin B_r$) that
\be
\langle D^s G, \chi_r \rangle = 0.
\ee 
\textbf{Case 2.} Let's now compute case $s-1 = 2\pi i k/\ln p$. We have the integral
\ba\label{eq324}
&\frac{1}{c_{s,p}\Gamma(\pi_{s+1})}\langle G, D^s\chi_r \rangle &= \lb \int_{x\notin B_r} \frac{\log_p|x-y|}{|x|^{s+1}}dx\rb\lb \int_{z\in B_r} dz \rb\\
&  - \lb\int_{x\in B_r} \log_p|x-y|dx\rb& \lb \int_{z\notin B_r} |z|^{-s-1} dz \rb.\nn
\ea
When $y\in B_r$, the integrals are
\ba
\int_{x\notin B_r} \frac{\log_p|x-y|}{|x|^{s+1}}dx &=& \int_{x\notin B_r} \frac{\log_p|x|}{|x|^{s+1}}dx = (1-p^{-1}) \sum_{i=r+1}^\infty i p^{-is}\\
&=& \frac{\left(p-1\right) \left((r+1) p^s-r\right) p^{-1-rs}}{\left(p^s-1\right)^2}, \\
\int_{x\in B_r} \log_p|x-y|dx &=& \int_{x\in B_r} \log_p|x'|dx' = (1-p^{-1}) \sum_{i=-\infty}^{r} i p^{i} \\
&=& p^{ r} \left(\frac{1}{1-p}+r\right).
\ea
Putting these together with Eqs. \eqref{eqq324} we obtain, when $y\in B_r$,
\be
\frac{1}{c_{s,p}\Gamma(\pi_{s+1})}\langle G, D^s\chi_r \rangle = \frac{p^{-1}+1}{p-1},
\ee
so that $\langle G, D^s\chi_r \rangle=1$.

When $y\notin B_r$, the integrals are
\ba
\int_{x\in B_r} \log_p|x-y|dx &=& (1-p^{-1}) \log_p |y| \sum_{i=-\infty}^{r} p^{i}\\ &=& p^{r} \log_p |y|, \\
\int_{x\notin B_r} \frac{\log_p|x-y|}{|x|^{s+1}}dx = &-& \int_{v_p(x) < v_p(y)} \frac{v_p(x)}{|x|^{s+1}} dx \\- \int_{v_p(x) = v_p(y)} \frac{v_p(x-y)}{|y|^{s+1}} dx
&  -& \int_{-r>v_p(x)>v_p(y)} \frac{v_p(y)}{|x|^{s+1}} dx. \nn
\ea
Denote $i \coloneqq v_p(x)$ and $b \coloneqq v_p(y)$. The first integral is 
\ba
\int_{v_p(x) < v_p(y)} \frac{v_p(x)}{|x|^{s+1}} dx &=& (1-p^{-1}) \sum_{i=-\infty}^{b-1} ip^{is}\\
&=& \frac{\left(1-p^{-1}\right) p^{b s} \left((b-1)p^s-b\right)}{\left(p^s-1\right)^2}.
\ea
To evaluate the second integral, we let $u = x-y$ and $j = v_p(u)$. We have this integral equal to
\ba
\int_{v_p(x) = v_p(y)} \frac{v_p(x-y)}{|y|^{s+1}} dx &=&\int_{j > b} \frac{j}{|u+y|^{s+1}} du + \int_{j = b} \frac{b}{p^{-(s+1)b}}\\
= bp^{sb}(1 - 2p^{-1})&+&(1-p^{-1})\sum_{j>b} j p^{(s+1)b - nj}\\
= bp^{sb}(1-p^{-1}) &+& \frac{p^{sb-1}}{1-p^{-1}}.
\ea
For the third integral, we have
\ba
\int_{-r>v_p(x)>v_p(y)} \frac{v_p(y)}{|x|^{s+1}} dx &=& \sum_{-r>i>b}bp^{is}(1-p^{-1})\\
&=& b(1 -p^{-1})\frac{p^{s(-r-1)}-p^{bs}}{1 - p^{-s}}.
\ea
Adding the three contributions, the result of this integral, using that $s = 1 + 2\pi ik/\ln p$, is simply
\be 
\int_{x\notin B_r} \frac{\log_p|x-y|}{|x|^{s+1}}dx = -b p^{ (-r-1)}.
\ee
Then plugging this back into Eq. \eqref{eq324}, we obtain that \eqref{eq324} equals $0$. This completes the proof.
\end{proof}

Below is the technical result we have used.

\begin{lemma}
For $y\in \mathbb{Q}_p$, $|y|\eqqcolon p^{r'}$ and $\Re(t)>-1$, we have the integral
\label{lemma2}
\be
\label{eq3p9}
\int_{|x|=|y|} |x-y|^{t} dx = \frac{\left(\left(p-2\right) p^{1+t}+1\right) p^{(r'-1)+r' t}}{p^{1+t}-1}.
\ee
\end{lemma}
\begin{proof}
We proceed by direct computation. We write $x$ as $x = u + y$ and divide this integral into two parts according to the $\mathbb{Q}_p$ valuation $v(u)$,
\be \int_{|x|=|y|} |x-y|^{t} dx = \int_{v(u) > v(y)} |x-y|^t dx  + \int_{\substack{v(u) = v(y)\\ v(x) = v(y)}} |x-y|^t dx. \ee
We evaluate the first integral, denoting $v(u) = i$ and using the volume factor for circle~$S_k$,
\ba 
\int_{v(u) > v(y)} |x-y|^t dx  &=& \int_{v(u) > v(y)} |u|^t du \\
 &=& (1 - p^{-1})\sum_{i \geq v(y)+1} p^{-(1+t)i} \nn \\
 &=& (1 -p^{-1})\frac{p^{-(1+t)(v(y)+1)}}{1 - p^{-(1+t)}}.\nn
\ea
The second integral is 
\be
\int_{\substack{v(u) = v(y) \\ v(x) = v(y)}} p^{-v(y)t} du = p^{-(1+t)v(y)}(1 - 2p^{-1}). 
\ee
Here measure factor $p^{-v(y)}(1 - 2p^{-1})$ obtained by subtracting the neighborhood of $u = -y$, since in that case $v(x) \neq v(y)$.

Summing up the two contributions, we obtain
\ba 
\int_{|x|=|y|} |x-y|^{t} dx &=& \frac{(p^{1+t}(p -2) + 1)p^{(r' -1) + r't}}{p^{1+t} -1},
\ea
as advertised.
\end{proof}

\begin{remark}
Result \eqref{mainresult1} is a consequence of translation invariance of the Haar measure. This can be understood as follows: It is possible to derive the result (along a much-simplified version of the steps above) at $y=0$, that is
\be
D_x^s G(x,0) =\delta(x).
\ee
By translation invariance this implies $D_x^s G(x-y,0) =\delta(x-y)$, i.e. $D_x^s G(x,y) =\delta(x,y)$. 
\end{remark}

\section{Equivalence of the two definitions of Vladimirov derivative operators for non-Archimedean completions of number fields}
\label{secnumberfield}

In this section we will establish the equivalence of the two derivative operators on general number fields for general Hecke characters. To introduce the terminology, denote the local field as $F$, a quasi-character as $\pi$, an additive character as $\chi$, the Schwartz space of compactly supported locally constant functions as $\{\psi\}$, and the generator of the unique prime ideal of $o_F$ by $p$. Furthermore, let $B_m$ be the set $B_{m} \coloneqq \{x \in F\,|\, v_p(x) \geq m\}$, denote the conductor of $\pi$ by $1 + B_{r}$, and let the conductor of $\chi$ be $B_{-d}$, which is the inverse~different. 

\begin{defn}
	The derivative operator obtained by Fourier transformation, denoted as $D_1$, is defined as 
	\be
	\label{D1ineq31}
	D_1 \psi(x') = \int_{F} \chi(-kx')\pi^{s}(k) \int_{F} \chi(kx) \psi(x) dx dk.
	\ee
\end{defn}
The Gamma function will be introduced as the coefficient associated with the Fourier transformation of a quasi-character.

\begin{remark}
The Fourier transformation of a quasi-character is another quasi-character. That is 
\be \label{ga}
C(s)\pi^{-s} =\int_{F} \chi(kx)\pi^{s}(x)|x|^{-1}dx, 
\ee
where $C$ is a certain constant. If $\pi$ is ramified with conductor $1 + B_{r}$, then
\be 
\label{gaa}
C(s) = \sum_{j \geq d+1} q^{sj} \int_{o^*_F} \chi(p^{-j}v) \tilde{\pi}(v) dv,
\ee
where $q  = Norm(p)$ and $\tilde{\pi}$ is the unitary part of $\pi$.
If $\pi$ is unramified, then 
\be
C(s) = \frac{q^{sd}(1 - q^{s-1})}{1 - q^{-s}}.
\ee
The integration \eqref{ga} is evaluated in the sense of summing over cycles, i.e summing over sets $\{x \in F\, |\, v_p(x) = n\}$ for $n \in \mathbb{Z}$. The rigorous proof of this is given in Tate's thesis. 
\end{remark} 
\begin{defn}
	We define the Gamma function as the constant deduced above. That is,
	\be
	\Gamma(\pi_s) = \sum_{j\geq d+1} q^{sj} \int_{o^*_F} \chi(p^{-j}v)\tilde{\pi}(v) dv
	\ee
	if $\pi$ is ramified, and
	\be
	\Gamma(\pi_s) = \frac{q^{sd}(1- q^{s-1})}{1 - q^{-s}}
	\ee
	if $\pi$ is unramified.
\end{defn}
\begin{defn}
	The derivative operator $D_2$ is defined as
	\be
	\label{D2ineq37}
	D_2 \psi(x) \coloneqq\Gamma(\pi_{s+1}) \int_{F} \frac{\psi(x') - \psi(x)}{\pi^{s}(x'-x)|x'-x|} dx'.
	\ee
\end{defn}

We want to establish the equivalence of the two derivative operators in Eqs. \eqref{D1ineq31} and \eqref{D2ineq37}.
To achieve this, we first prove the following lemma.

\begin{lemma}
	We have the integral
	\be
	\int_{o^*_F} \chi(p^{-j}v)\tilde{\pi}(v) dv = 0,
	\ee
	when $ j \neq d+ r$ and $j \geq d+ 1$.
\end{lemma}
\begin{proof}
	For $ d+ 1 \leq j \leq d+r-1$, we pick $1 + up^{r-1}$ such that $\tilde{\pi}(1 + up^{r-1}) \neq 1$. Therefore, the integral equals
	\ba
		\int_{o^*_F} \chi(p^{-j}v)\tilde{\pi}(v) dv &=&\int_{o^*_F} \chi(p^{-j}v(1+up^{r-1}))\tilde{\pi}(v(1 +up^{r-1})) dv\\
	&=&\tilde{\pi}(1 + up^{r-1}) \int_{o^*_F} \chi(p^{-j}v) \tilde{\pi}(v) dv. \nn
	\ea
	This implies that the integral is zero.
	In the case  $j \geq d + r + 1$, pick $ up^{r}$ such that $\chi(up^{r-j}) \neq 1$. Then
	\ba
		\int_{o^*_F} \chi(p^{-j}v)\tilde{\pi}(v) dv &=& \int_{o^*_F} \chi(p^{-j}(v +up^{r}))\tilde{\pi}(v) dv\\
	&=&\chi(p^{r-j}u)\int_{B_{-m}} \chi(p^{-j}v)\tilde{\pi}(v)dv.
	\ea
	This means that the integral vanishes.
\end{proof}
\begin{lemma}\label{invariant}
	If  $D_1 \phi(x) = D_2 \phi(x)$, where $\phi(x)$ is the characteristic function of the ball centered at $0$ with radius $1$, i.e $B_{1}$, then $D_1 \phi'(x) = D_2\phi'(x)$, where $\phi'(x)$ is the characteristic function of the ball centered at $b \in F$ with arbitrary radius $m$, with $m \in \mathbb{Z}$. Note that the radius is given in terms of valuation.
\end{lemma}
\begin{proof}
	Let $y = (x - b)/p^{m}$ and $y' = (x'- b)/p^{m}$, then $D_1 \phi'(x)$ becomes
	\ba
	 D_1 \phi'(x') &=& \int_{F} \chi(-k(p^{m}y'+b)) \pi^s(k) \times\\
	& &\times\int_{F} \chi(k(p^{m}y+b))\phi'(p^{m}y + b) dp^{m}y dk\nonumber\\
	&=& \int_{F} \chi(-uy') \pi^s(u/p^{m}) \int_{F} \chi(uy)\phi(y) dydu\\
	&=& \pi^{-s}(p^{m}) \int_{F} \chi(-uy') \pi^s(u) \int_{F} \chi(uy)\phi(y) dydu,
	\ea
where in the second equality we have denoted $u \coloneqq p^{m}k$. For $D_2 \phi'(x)$, we have
	\ba
	D_2 \phi'(x')&=& \Gamma(\pi_{s+1}) \int_{F} \frac{\phi'(p^{m}y+b) - \phi'(p^{m}y'+ b)}{\pi^s(p^{m}y- p^{m}y')|p^{m}y-p^{m}y'|} dp^{m}y\\
	&=& \pi^{-s}(p^{m}) \Gamma(\pi_{s+1}) \int_{F} \frac{\phi(y) - \phi(y')}{\pi^s(y - y')|y-y'|} dy.
	\ea
	Therefore we obtain $D_1 \phi'(x) = D_2 \phi'(x)$, by our assumptions.
\end{proof}

\begin{thm}
	Derivative operators $D_1$ and $D_2$ are equivalent in the sense that for any compactly supported locally constant test function $\psi(x)$, we have
	\be
	D_1 \psi(x) = D_2 \psi(x).
	\ee
\end{thm}

\begin{proof}
	We first consider the characteristic function of a ball $B_{1}$. That is, $\phi(x)$ is defined~as 
	\begin{equation}
     \phi(x) =
	\begin{cases}
	1 & \text{if $v_p(x) \geq 1$}\\
	0 & \text{otherwise}
	\end{cases}.
	\end{equation}
	Then we evaluate the first definition of the derivative operator on this function:
	\ba 
	D_1 \phi(x') &=& \int_{F} \chi(-kx')\pi^s(k) \int_F \chi(kx) \phi(x)dxdk\\
	&=& q^{-1} \int_{p^{-1}D^{-1}_F} \chi(-kx)\pi^s(k)dk.
	\ea
	Note that the inverse different can be written as $B_{-d}$, so that
	\be
	D_1 \phi(x') = q^{-1} \int_{B_{-1-d}} \chi(-kx')\pi^s(k)dk,
	\ee
	and letting $u = -k$ we have
	\be
	D_1 \phi(x') = q^{-1} \int_{B_{-1-d}} \chi(ux')\tilde{\pi}(-u)|u|^s du.
	\ee
	Now let $v = ux'$, we obtain
	\ba
	D_1 \phi(x') &=& q^{-1} \tilde{\pi}^{-1}(x')|x'|^{-s-1}\\ &\times& \int_{B_{-1-d+c}} \chi(v) \tilde{\pi}(-v)|v|^{s+1} d^*v\nonumber\\
	\label{eq324forfirstdefder}
	 &=& q^{-1} \tilde{\pi}^{-1}(x')|x'|^{-s-1} \\&\times&\sum_{i\geq -1-d+c} q^{-(s+1)i} \int_{o^*_F} \chi(vp^i)\tilde{\pi}(-v)dv \nonumber.
	\ea
	If $\pi$ is ramified and $x' \in B_{1}$, this equals $0$. If instead $\pi$ is unramified and $x' \in B_{1}$, this expression gives
	\be
	\label{eqD1rami}
	D_1 \phi(x')= \frac{(1-q^{-1})q^{s}q^{(s+1)d}}{1- q^{-(s+1)}}.
	\ee
	If $\pi$ is unramified and $x'\notin B_{1}$, we have the derivative equal to 
	\be
	D_1 \phi(x')= \Gamma(\pi_{s+1})q^{-1}|x'|^{-s-1}.
	\ee
	Finally, if $\pi$ is ramified and $x'\notin B_{1}$, we have
	\ba
	D_1 \phi(x')&=& q^{-1}\tilde{\pi}^{-1}(-u)|x'|^{-s-1}\\ &\times&\sum_{d+1 \leq j \leq d+1-c} q^{(s+1)j} \int_{o^*_F} \chi(p^{-j}v)\tilde{\pi}(v)dv \nonumber.
	\ea
	Next, we proceed to calculate the second definition of the derivative operator. Notice that if $\pi$ is ramified then the Gamma function in this case is equal to
	\be 
	\Gamma(\pi_s) = \sum_{j\geq d+1}q^{sj} \int_{o^*_F} \psi(p^{-j}v)\tilde{\chi}(v)dv.
	\ee
	If $\pi$ is unramified, we have the Gamma function
	\be
	\Gamma(\pi_s)= \frac{q^{sd}(1 - q^{s-1})}{1 - q^{-s}}.
	\ee
	From the second definition of the derivative operator,
	\be
	D_2 \phi(x') = \Gamma(\pi_{s+1}) \int_F (\phi(x) - \phi(x'))\pi^{-s}(x-x')|x-x'|^{-1}dx,
	\ee
	which with the change of variables $u = x-x'$ becomes
	\be
	D_2 \phi(x') = \Gamma(\pi_{s+1}) \int_F (\phi(u+x') - \phi(x')) \pi^{-s}(u)|u|^{-1}du.
	\ee
	We now compute this expression in different cases. For $x' \in B_{1}$, we have 
	\ba
	D_2 \phi(x') &=& - \Gamma(\pi_{s+1}) \int_{B^{c}_{1}} \pi^{-s}(u)|u|^{-1} du\\
	&=&-\Gamma(\pi_{s+1}) \int_{B^{c}_{1}} |u|^{-s} |u|^{-1}du\\
	&=& -\Gamma(\pi_{s+1})\sum_{i <1} q^{si} \int_{o^*_F} \tilde{\pi}(u)d^*u.
	\ea
	If $\pi$ is ramified this will equal $0$. Otherwise,
	\be
	D_2 \phi(x')= \frac{(1-q^{-1})q^{s}q^{(s+1)d}}{1- q^{-(s+1)}}.
	\ee
	This matches result \eqref{eqD1rami} for the first definition of the derivative operator. Next, in the case $x'\notin B_{1}$ we have
	\be
	D_2 \phi(x')=\Gamma(\pi_{s+1})\int_{B_{1}-x'}\pi^{-s}(u)|u|^{-1}du.
	\ee
	If $\pi$ is unramified then this equals
	\be
	D_2 \phi(x')= \Gamma(\pi_{s+1}) |x'|^{-s-1} q^{-1},
	\ee
	and if $\pi$ is ramified it equals
	
	\ba
	    D_2 \phi(x') &=& \Gamma(\pi_{s+1})|x'|^{-s-1}\\ &\times&\int_{B_{1}-x'} \tilde{\pi}^{-1}(u) du \nonumber \\
	&=& |x'|^{-s-1}\int_{B_{1} -x'} \tilde{\pi}^{-1}(u)du \\&\times& \sum_{j \geq d+1} q^{(s+1)j} \int_{o^*_F}\chi(p^{-j}v)\tilde{\pi}(v) dv\nonumber. 
	\label{a}
	\ea
	
	We want to show this is equal to Eq. \eqref{eq324forfirstdefder} for the first definition of the derivative~operator,
	
	\ba
	\label{b}
	D_1 \phi(x')&=& \tilde{\pi}^{-1}(-x')q^{-1}|x'|^{-s-1}\\&\times&\sum_{d+1 \leq j \leq d+1-c} q^{(s+1)j}\int_{o^*_F} \chi(p^{-j}v)\tilde{\pi}(v)dv\nonumber.
	\ea
	We divide the computation into two subcases. First let's analyze the case $v_p(x') >1 - r$, where $r$ is the degree of ramification of $\tilde{\pi}$. Then consider the integral 
	\ba 
	\label{c}
	\int_{B_{1} -x'} \tilde{\pi}(u) du &=& \int_{B_{1}} \tilde{\pi} (u -x') du \\
	&=& \tilde{\pi}(-x') \int_{B_{1}} \tilde{\pi}(1 - u/x') du.
	\ea
	Pick $1 - v$ from $o^*_F$ such that $\tilde{\pi}(1 -v)$ is nontrivial. We change the variable to $u' = (1-v)^{-1}(u-vx')$, and we have
	\be
	\int_{B_{1} -x'} \tilde{\pi}(u) du = \tilde{\pi}(-x') \int_{B_{1}} \tilde{\pi}((1 - v)(1 - u'/x')) du'.
	\ee
	This implies that Eq. \eqref{c} vanishes. From the lemma above, we have that Eqs. \eqref{a}, \eqref{b} are both equal to $0$, since $r > 1 -v_p(x')$, and the integral in the lemma has nonzero value only at $j = d+ r$.
	For the second case, i.e. $v_P(x') \leq 1 -r$, Eq.~\eqref{c}~is
	\be
	\int_{B_{1} -x'} \tilde{\pi}(u) du = \tilde{\pi}(-x') Vol(B_{1}) = \tilde{\pi}(-x')q^{-1}.
	\ee
	Therefore the expressions \eqref{a} and \eqref{b} are equal, since they only have nonzero value at $j = d+r$, and are identical in that case.
	We thus conclude that the two definitions of the Vladimirov derivative operators are equivalent when evaluating on $\phi(x)$. Since every test function $\psi(x)$ can be written as linear combinations of $\phi'(x)$, that is of characteristic functions of balls with arbitrary center and radius, from the lemma \ref{invariant} above, the two definitions must be equal on test function space.
	
\end{proof}

\section{Field theory on a number field}
\label{cftnf}
Fix a number field $K$ and a Hecke character $\chi$. For any place $\nv$ of $K$, we have the pseudo-differential operator $\partial^{\chi_{\nv}}\coloneqq D_2$ on the Schwartz-Bruhat space $S(K_{\nv})$ as defined in Eq. \eqref{D2ineq37}, with respect to $\chi_{\nv}$, the component of $\chi$ at $\nv$. Consider the action 

\begin{equation}
S=\int_{K_{\nu}}\phi(x) \partial^{\chi_{\nv}}\phi(x)dx.
\end{equation}

By the standard path-integral argument of quantum field theory (see e.g. \cite{Gubser:2017qed}), the two-point function of this free theory is a Green's function $G_{\nv}$ of $D_2$. Since $D_1=D_2$, where $D_1=\FF^{-1}\chi_{\nv}\FF$, $G_{\nv}$ is then the Green's function of $D_1$, which can be obtained via Tate's thesis: in Tate's thesis, the Zeta distribution  $\chi_{\nv}\eqqcolon|\cdot|^s\tilde{\chi}_{\nv}$ as an operator on the Schwartz space is analytically continued to a distribution-valued meromorphic function of $s$ on the whole complex plane, via the functional equation of local Zeta integrals. In particular, if the function  $b(t)\coloneqq|\cdot|^t\tilde{\chi}_{\nv}^{-1}$ is analytic at $-s$, then the evaluation at $-s$ of $b(t)$ gives a canonical distributional inverse of $\chi_{\nv}$. If $b(t)$ has a pole at $-s$, then the constant term of the Laurent series expansion of $b(t)$ at $t=-s$ gives a distributional inverse of $\chi_{\nv}$. We don't discuss such poles in detail in this paper. Denote such an inverse distribution by $\chi_{\nv}^{-1}$. Then $\FF^{-1}\chi_{\nv}^{-1}$ is a Green's function of $D_1$. Due to the positive definiteness of $D_1$, any other Green's function can only differ from this one by adding a constant, but by the standard path-integral integral argument, $\FF^{-1}\chi_\nu^{-1}$ (i.e. Eq. \eqref{Gishere} in the case $F=\Qp$) is precisely the two-point-function. Now, when $\Re(s)>1$, by the local functional equation of Zeta integrals, $\FF^{-1}\chi_{\nv}^{-1}$ is equal to the distribution $\gamma \chi_{\nv}|\cdot|^{-1}$, where $\gamma=\gamma(\chi_{\nv},\psi,dx)$ is a constant (the local Gamma factor) depending on $\chi_{\nv}$, the choice of the additive character $\psi$, and the Haar measure $dx$. Furthermore, the global functional equation of Zeta integrals says that these local Gamma factors multiply to $1$ in the sense of analytic continuation. Now since the local quasi-characters $\chi_{\nv}|\cdot|^{-1}$ also multiply to $1$, we deduce that the two point functions multiply to $1$, i.e.
\begin{equation}
\label{prod}
    \prod_{\nv\leq\infty} G_{\nv}=1.
\end{equation}
When $\chi_{\nv}$ is ramified (i.e. when its unitary part is nontrivial), the local Gamma factor is regular at $s=1$, and the above equality holds also at $s=1$. 

When $\chi_{\nv}$ is unramified, the local Gamma factor has a pole at $s=1$. In this case, the local Green's function has to be worked out by picking the constant term as described above. Furthermore, taking the derivative of the product formula \eqref{prod} with respect to $s$, and evaluating the resulting formula at $1$, one obtains an adelic relation of the local Green's functions at $s=1$.

\section{Green's functions for the maximum norm on $\mathbb{Q}_p^n$}
\label{secQpn}

In this section we show that the Green's functions derived in the general framework of Section~\ref{sec21} also apply to derivatives defined from quasi-characters corresponding to the maximum norm on $\mathbb{Q}_p^n$. Our discussion can be interpreted as a kind of ``higher-dimensional'' ($n>1$) extension of the standard treatment for $n=1$ found in Section~\ref{sec2}, and in texts such as \cite{GGIP}, \cite{vvzbook}. It is interesting that many of the results in Section \ref{sec21} will apply (as we will show), even though $\mathbb{Q}_p^n$ is not a field. 

In this section we will only present the Green's function for the integral representation on $\mathbb{Q}_p^n$ of type $D_2$ of the Vladimirov derivative. These Green's functions have been studied before in the literature, in a different manner, e.g. in \cite{zunigabook}; because of this our presentation here will be brief. One can also see Theorem 137 in \cite{zunigabook} for a relation between the integral representation $D_2$ and the Fourier transform representation $D_1$ of the Vladimirov derivative in this case, which we will not discuss here.\footnote{When $n=3$ the two representations of the Vladimirov derivative have also been discussed in \cite{Abdesselam:2013zia}; we thank the anonymous referee for pointing out this paper.}

It is also worth mentioning that for $n>1$ the choice of maximum norm as local component does not correspond to a Hecke character, therefore there is no immediate adelic interpretation of the computations in this section. However, choices of Hecke characters can still be made for specific values of $n$ (for instance in relation to quaternion algebras when $n=4$), so that the framework in Section \ref{sec21} applies to the local components, and adelic products hold; we will not explore this in the present paper.

We will employ the notation $D_i$, $i\in\mathbb{Z}$, for the disk
\be
D_i \coloneqq\lcb x\in\mathbb{Q}_p\ |\ v_p(x) > i \rcb
\ee
of volume
\be
\label{volD}
\mrm{vol}\lb D_i\rb=p^{-i-1},
\ee
with $v_p$ the usual $p$-adic valuation. We use Greek indices $\alpha,\beta,\dots\in\lcb1,\dots,n\rcb$ to denote the various $\Qp$ components in~$\mathbb{Q}_p^n$.

We now introduce the maximum norm $|\cdot|_n$ on $\mathbb{Q}_p^n$.

\begin{defn}
We define the $p$-adic norm $|\cdot|_n$ on $\mathbb{Q}^n_p$ as the maximum norm, which means
\be
\label{defeq1}
|x|_n \coloneqq \max_\alpha |x_\alpha|_p
\ee
for $x\in\mathbb{Q}_p^n$, with $|\cdot|_p$ the usual $p$-adic norm on $\Qp$. Then correspondingly we define the valuation $v_n(x)$ as
\be
\label{defeq2}
v_n\lb x \rb = \min_\alpha v_p(x_\alpha).
\ee
\end{defn}


\begin{lemma}
\label{lemmabasic}
The definitions in Eqs. \eqref{defeq1} and \eqref{defeq2} satisfy the properties required of norm and valuation respectively.
\end{lemma}
\begin{proof}
Immediate.
\end{proof}
\begin{remark}
The maximal subgroup of $\mathrm{GL}(n,\mathbb{Q}_p)$ preserving the maximum norm \eqref{defeq1} is its maximal compact subgroup $K=\mathrm{GL}(n,\mathbb{Z}_p)$: Let $A\in \mathrm{GL}(n,\mathbb{Z}_p)$, and let $y=Ax$, where $x\in\mathbb{Q}_p^n$ is the column vector $[x_1,x_2,...,x_n]^T$. For any $i$, $y_i=\sum_{j=1}^nA_{ij}x_j$, and therefore $|y_i|_p\leq max_j|A_{ij}x_j|_p\leq max_j|x_j|_p$, as $|A_{ij}|_p\leq 1$ for any $i,j$. Thus $|y|_n=max_i|y_i|_p\leq max_j|x_j|_p=|x|_n$. On the other hand, since $A\in \mathrm{GL}(n,\mathbb{Z}_p)$, $A^{-1}\in \mathrm{GL}(n,\mathbb{Z}_p)$. Thus the same argument switching the roles of $x$ and $y$ shows that $|x|_n\leq |y|_n$. So we have $|x|_n=|y|_n$, i.e. $\mathrm{GL}(n,\mathbb{Z}_p)$ preserves the maximum norm. Lastly, the Iwasawa decomposition applied to elements in $\mathrm{GL}(n,\mathbb{Q}_p)$ easily shows that $\mathrm{GL}(n,\mathbb{Z}_p)$ is the maximum subgroup of $\mathrm{GL}(n,\mathbb{Q}_p)$ that preserves the maximum norm.

This is in analogy with the Archimedean case, where the maximal subgroup preserving the quadratic kinetic term is the maximal compact subgroup of $\mathrm{GL}(n,\mathbb{R})$ (or $\mathrm{GL}(n,\mathbb{C})$), i.e. the orthogonal group (or the unitary group). 
\end{remark}
\begin{remark}
For $K=\mathbb{Q}_p^n$, the Haar modulus $\Delta$ is given by
\be
\Delta\lb x=\lb x_1,\dots,x_n \rb \rb = \prod_{\alpha=1}^n |x_\alpha|_p.
\ee
\end{remark}

\begin{lemma}
\label{lemma1}
The volume of the $p$-adic circle
\be
S_{n,i}\coloneqq \lcb x\in\mathbb{Q}_p^n\ |\ v(x)=i \rcb,
\ee
with $i\in\mathbb{Z}$, is
\be
\label{volSni}
\mrm{vol}\lb S_{n,i} \rb = (1-p^{-n})p^{-ni}.
\ee
\end{lemma}
\begin{proof}
Circle $S_{n,i}$ consists of all possible distributions of $n$ indices $\alpha$ such that $v(x_\alpha) \geq i$ for all $\alpha$ and $v(x_\alpha)=i$ for at least one index $1\leq\alpha\leq n$. Thus $S_{n,i}$ can be written as
\be
S_{n,i} = \bigcup_{u=1}^n \bigcup_{\mathcal{I}_u} \lb \prod_{\alpha\in \mathcal{I}_u} S_i^{(\alpha)} \rb  \otimes\lb \prod_{\beta\in \{1,\dots,n\}-\mathcal{I}_u } D_i^{(\beta)} \rb,
\ee
where $\mathcal{I}_u$ is a set of indices of cardinality $u$. Then the volume is
\be
\mrm{vol}\lb S_{n,i}\rb = \sum_{u=1}^n \sum_{\mathcal{I}_u} \lb \prod_{\alpha\in \mathcal{I}_u} \mrm{vol}\lb S_i^{(\alpha)} \rb  \rb\lb \prod_{\beta\in \{1,\dots,n\}-\mathcal{I}_u} \mrm{vol} \lb D_i^{(\beta)} \rb \rb.
\ee
There are $\binom{n}{u}$ ways to choose $\mathcal{I}_u$, so that using Eqs. \eqref{volS}, \eqref{volD} the volume becomes
\ba
\mrm{vol}\lb S_{n,i}\rb &=& \sum_{u=1}^n \binom{n}{u} \lb p^{-i}-p^{-i-1} \rb^{u}  \lb  p^{-i-1} \rb^{n-u} \\
&=& \lb 1-p^{-n} \rb p^{-ni}.
\ea
\end{proof}

\begin{defn}
Let $\pi^\mrm{max}_s(x)\coloneqq |x|_n^s$ and ${\mathbb{Q}_p^n}^\times\coloneqq\{x\in\mathbb{Q}_p^n \ |\ x_\alpha\neq0,\forall\ \alpha=1,\dots,n\}$.
\end{defn}

We have that $\pi_s$ is a quasi-character on ${\mathbb{Q}_p^n}^\times$, meaning that $\pi^\mrm{max}_s(xy)=\pi^\mrm{max}_s(x)\pi^\mrm{max}_s(y)$, with component-wise multiplication.

We now compute the explicit values of the Gamma function associated to the maximum norm on $\mathbb{Q}_p^n$. 

\begin{lemma}
We have the integral
\be
\int_{S_i} \chi(x) = \begin{cases} \mrm{vol} \lb S_i\rb \quad \mrm{if} \quad i > -1 \\
-1 \hspace{0.83cm} \quad \mrm{if} \quad i=-1\\
0 \hspace{1.15cm} \quad \mrm{if}\quad i<-1
\end{cases}.
\ee
\end{lemma}
\begin{proof}
This follows from direct computation.
\end{proof}

\begin{lemma}
We have the integral
\be
\int_{D_i} \chi(x) = \begin{cases}
p^{-1-i} \quad \mrm{if}\quad i\geq-1\\
0 \quad \hspace{0.75cm}\mrm{if}\quad i<-1
\end{cases}.
\ee
\end{lemma}
\begin{proof}
Direct computation using
\be
\int_{D_i}\chi(x)= \sum_{j=i+1}^\infty \int_{S_j} \chi(x)  
\ee
gives the desired answer.
\end{proof}

We now introduce the Gamma function associated to the character $\pi^\mrm{max}_s$.

\begin{lemma}
The value of the $\Gamma$ function for the maximum norm is
\be
\Gamma\lb \pi^\mrm{max}_s \rb = \frac{1-p^{s-n}}{1-p^{-s}}.
\ee
\end{lemma}

\begin{proof}
By direct computation we have
\ba
&&\Gamma\lb \Delta \pi_s^\mrm{max} \rb = \int_{{\mathbb{Q}^n_p}^\times} \pi_s^\mrm{max}\lb x \rb e^{2\pi i\{x\}}dx \\&=& \sum_{i=-\infty}^\infty \int_{S_{n,i}} \pi_s^\mrm{max}\lb x \rb e^{2\pi i\{x\}}dx\nonumber \\
&=& \sum_{i=-1}^\infty\sum_{u=1}^n \sum_{\mathcal{I}_u} p^{-is} \lb \prod_{\alpha\in \mathcal{I}_u} \int_{S_i^{(\alpha)}} e^{2\pi i\{x_\alpha\}}dx_\alpha \rb\times\\
& &\times\lb \prod_{\beta\in \{1,\dots,n\}-\mathcal{I}_u} \int_{D_i^{(\beta)}} e^{2\pi i\{x_\beta\}}dx_\beta \rb \nn\\
&=& \sum_{u=1}^n \binom{n}{u} \lsb (-1)^u p^{s} +\sum_{i=0}^\infty p^{-is} \lb p^{-i}-p^{-i-1} \rb^u p^{-(1+i)(n-u)} \rsb \nn \\
\label{eq320}
&=& \frac{1-p^s}{1-p^{-n-s}},
\ea
where the sum converges when $\Re\lb s+n \rb>0$.
Using the functional equation \eqref{gammaFUN}, and that $\lb\pi^\mrm{max}_s\rb^{-1}=\pi^\mrm{max}_{-s}$, we obtain
\be
\label{Gammasn}
\Gamma\lb \pi^\mrm{max}_s \rb = \frac{1-p^{s-n}}{1-p^{-s}},
\ee
as advertised.
\end{proof}

\begin{remark}
Noting that $\Gamma\lb \Delta \pi_s^\mrm{max} \rb=\Gamma\lb \pi_{s+n}^\mrm{max} \rb$, the functional equation can be recast~as
\be
\Gamma\lb \pi^\mathrm{max}_{n-s} \rb \Gamma\lb \pi^\mathrm{max}_{s} \rb = 1.
\ee
\end{remark}

\begin{remark}
The unique zero of the Gamma function is at $\pi^\mrm{max}=\pi_n^\mrm{max}$, and the unique pole is at $\pi^\mrm{max}=\pi_0^\mrm{max}$.
\end{remark}

Let's now introduce the integral representation of the Vladimirov derivative, and the Green's function. The discussion here and in Appendix \ref{appA} will closely follow that in Section \ref{sec2}.
\begin{defn}
For parameter $s\in\mathbb{C}$ with $\Re\lb s\rb>0$, and functions $\psi:\mathbb{Q}^n_p\to\mathbb{C}$ in the Schwartz space, the Vladimirov derivative at $x\in\mathbb{Q}_p^n$ is given by
\be
\label{thisderrr}
D_x^s \psi(x) \coloneqq \Gamma(\pi^\mrm{max}_{s+n}) \int_{\mathbb{Q}_p^n} \frac{\psi\lb x' \rb - \psi\lb x \rb }{|x'-x|^{s+n}_n}dx'.
\ee
\end{defn}

\begin{lemma}
For $\Re\lb s\rb > 0$, the Vladimirov derivative is a self-adjoint operator from the Schwartz space of compactly supported locally constant functions, to a bigger function space inside the space of continuous functions on $\mathbb{Q}_p^n$.
\end{lemma}
\begin{proof}
Immediate, using definition \eqref{thisderrr}, as in the proof of Lemma \ref{lemmasadj}.
\end{proof}

\begin{defn}
For $x,y\in\mathbb{Q}_p^n$ and $s\in\mathbb{C}$ we introduce a function $G(x,y)$ as
\be
G\lb x,y \rb\coloneqq \begin{cases}
c_{n,s,p}|x-y|_n^{s-n} \hspace{0.87cm} \mrm{if} \quad s-n\neq 2\pi i k/\ln p, \quad k\in\mathbb{Z} \\
c_{n,s,p}\log_p |x-y|_n \quad \mrm{if} \quad s-n =  2\pi i k/\ln p, \quad k\in\mathbb{Z}
\end{cases},
\ee
where constant $c_{n,s,p}$ is given by
\be
c_{n,s,p} = \begin{cases}
\Gamma\lb \pi^\mrm{max}_{n-s} \rb \hspace{0.96cm} \mrm{if}\quad  s-n \neq 2\pi i k/\ln p, \quad k\in\mathbb{Z}\\
-\lb 1-p^{-n}\rb \quad \mrm{if}\quad  s-n = 2\pi i k/\ln p, \quad k\in\mathbb{Z}
\end{cases}.
\ee
\end{defn}

\begin{thm} 
\label{hereisthm11111}
For $x\neq y$ and $\Re(s)>0$ we have that 
\be
D^s_x G(x,y) = 0.
\ee
\end{thm}

\begin{thm} 
\label{mainresultz11111}
When $\Re(s)>0$, $G(x,y)$ is the Green's function for the Vladimirov derivative $D_x^s$ on $\mathbb{Q}_p^n$, that is
\be
D_x^s G\lb x,y \rb = \delta \lb x-y \rb,
\ee
acting on the Schwartz space of compactly supported locally constant functions.
\end{thm}

The proofs of Theorems \ref{hereisthm11111} and \ref{mainresultz11111} are given in Appendix \ref{appA} below.

\section{Outlook: Berkovich spaces, renormalization, and functional equations}

In this paper we have explained how Tate's thesis is equivalent to the fact that the free Archimedean propagator in field theory can be reconstructed from the $p$-adic propagators, via an Euler product formula. There exists another appearance of a product formula in physics, in the case of amplitudes: Archimedean tree level four-point scattering amplitudes can be reconstructed from $p$-adic amplitudes \cite{Ruelle:1989dg}. In light of our discussion in this paper, this is again equivalent to a special case of Tate's thesis. It is remarkable that Tate's thesis makes appearances in physics in this manner. It is possible to give intuition for why Tate's thesis appears in the reconstruction of two-point functions and amplitudes, along the lines discussed in paper \cite{HMS}. Paper \cite{HMS} interprets the Euler products for the two-point functions and amplitudes as an equation of motion in Berkovich space. This equation of motion encodes a kind of renormalization group flow, for objects which scale trivially under changes in the energy scale. Thus, Tate's thesis is simply the equation of motion in Berkovich space for these types of objects. It will be exciting to explore these ideas further.

\subsection{$L$-functions and physics}

The results in the present paper conjecturally suggest that $L$-functions can be associated to $p$-adic and Archimedean free theories, and in particular that Euler products of the free propagators are equivalent to functional equations. So far, the evidence for this comes from the degree $1$ $L$-functions studied in this paper and (for the more general proposal) from the connections between the Zeta function and $p$-adic string theory noted in \cite{Huang:2019nog}. However, it is possible to tentatively flesh out things a little bit more.

Our proposal is a (for now conjectural) \emph{recipe} for deriving functional equations: write down a family of $p$-adic theories, deduce their Archimedean limit, and the functional equation for the function associated to the family of theories follows by demanding that the Archimedean propagator is the inverse product of the $p$-adic~propagators. Currently, it is unclear how the Archimedean counterpart to the $p$-adic theories should be determined in a rigorous manner, however it is possible to \emph{heuristically} guess what it should be, at least in some cases. For this, remember that the local component $\pi$ of the Hecke character entering the Vladimirov derivative in the action is the same as the character entering the Green's function \eqref{sketchG}, and furthermore this Green's function has the form
\be
\label{eqGfct}
G_\pi\lb x-y \rb = \Gamma\lb \Delta \pi^{-1} \rb \Delta^{-1} \pi (x-y).
\ee
In this expression the nontrivial information on the functional equation is encoded only in the Gamma prefactor, since the $\Delta^{-1}\pi\lb x-y \rb$ spatial dependence cancels in the Euler product. Thus, the Archimedean theory should be picked precisely such that this spatial dependence cancels. This restriction significantly constrains the form of the theories, as we explain below.

Although in the present paper we have dealt only with Archimedean theories defined on $\mathbb{R}$ or $\mathbb{C}$, it is possible to sketch out how our proposal adapts to other cases, such as those of scalar theories defined on $\mathbb{R}^4$. This is a significant departure from Tate's thesis, and also perhaps one of the more physically interesting cases. The key point is that the spatial dependence of the Green's function should cancel, and this can be ensured by making use of quaternion algebras. Consider $ x \coloneqq\lb x_1, \dots, x_4 \rb \in \mathbb{Q}_p^4$, and define
\ba
\pi^s_{(p)}\lb x \rb &\coloneqq& \left| \det \begin{pmatrix}
	x_1 & -x_2 & -x_3 & -x_4\\
	x_2 & x_1 & -x_4 & x_3\\
	x_3 & x_4 & x_1 & -x_2\\
	x_4 & -x_3 & x_2 & x_1\\
\end{pmatrix} \right|_p^\frac{s}{4} \\
&=& \left| x_1^2 + x_2^2 + x_3^2 + x_4^2 \right|^\frac{s}{2}_p.
\ea
$\pi^s_{(p)}$ inherits the multiplicative structure
\be
\pi^s_{(p)}(xy) = \pi^s_{(p)}(x) \pi^s_{(p)}(y)
\ee
from the quaternion division algebra, and so can be used as the character with respect to which a Vladimirov derivative is defined. Furthermore, one can define an Archimedean character as
\be
\pi^s_{(\infty)} \coloneqq \left| x_1^2 + x_2^2 + x_3^2 + x_4^2 \right|^\frac{s}{2}_\infty,
\ee
which is multiplicative under the quaternion algebra defined over the reals, and such that
\be
\pi_{(\infty)}^s\prod_p \pi_{(p)}^s = 1.
\ee
On the Archimedean side, the Green's function given by this multiplicative character corresponds to the usual Laplacian when $s=2$. The spatial dependence on $x-y$ crucially cancels in the Euler product of the Green's functions \eqref{eqGfct}. Furthermore, the machinery sketched in Section \ref{sec21} formally applies, so it is natural to guess  one also has an adelic product formula for the Green's functions.

More generally, one can consider theories defined not only over $\mathbb{Q}_p^n$, but also over other objects such as the Bruhat-Tits tree $T_p$. This is natural from the point of view of physics, because theories on $\mathbb{Q}_p$ are related to those on $T_p$ by the holographic duality. In \cite{Huang:2019nog} a connection between scalar theories and the Zeta function was noted, however it is possible to expand the scope of that analysis to better understand the relation between the content of the physical theories and the properties of the $L$-functions, which will be reported in future work. This analysis may also be applicable to Ising models, by changing the target space of the theory.

\section*{Acknowledgments}

We thank C. Jepsen, O. Offen, M. M. Nastasescu, and  W. A. Z\'u\~niga-Galindo for useful discussions. The work of A.H. and S.-T. Y. was supported in part by a grant from the Simons Foundation in Homological Mirror Symmetry. The work of A.~H., B.~S., and X.~Z. was supported in part by a grant from the Brandeis University Provost Office. B.~S. was supported in part by the U.S. Department of Energy under grant DE-SC-0009987, and by the Simons Foundation through the It from Qubit Simons Collaboration on Quantum Fields, Gravity and Information.

\appendix 

\section{Technical details on the Green's function computation for the maximum norm on $\mathbb{Q}_p^n$}
\label{appA}

In this appendix we give the proofs of theorems \ref{hereisthm11111} and \ref{mainresultz11111} in Section \ref{secQpn}. The discussions (and the proofs) mirror those in Section \ref{sec2} precisely. To not encumber notation, we will denote the $|\cdot|_n$ max norm on $\mathbb{Q}_p^n$ by $|\cdot|$ in this appendix, and the valuation $v_n$ by $v$.
\begin{thm} 
\label{appAhereisthm1}
For $x\neq y$ and $\Re(s)>0$ we have that 
\be
D^s_x G(x,y) = 0.
\ee
\end{thm}
\begin{proof} The proof proceeds by direct computation. In the case $s-n\neq 2\pi i k/\ln p$ we need to evaluate the integral
\be
I_n(s,v) \coloneqq \int_{\mathbb{Q}_p^n} \frac{|z-y|^{s-n} - |x-y|^{s-n}}{|z-x|^{s+n}} dz = \int_{\mathbb{Q}_p^n} \frac{|u|^{s-n} - |v|^{s-n}}{|u-v|^{s+n}} du,
\ee
where $u\coloneqq z-y$ and $v\coloneqq x-y$. Assuming $v\neq 0$, we have
\ba
\label{appAeq416}
& &I_n(s,v) = \int_{\mathbb{Q}_p^n} \frac{|u|^{s-n} - |v|^{s-n}}{|u-v|^{s+n}} du \\
&=& \int_{|u|<|v|} \frac{|u|^{s-n} - |v|^{s-n}}{|v|^{s+n}} du + \int_{|u|>|v|} \frac{|u|^{s-n} - |v|^{s-n}}{|u|^{s+n}} du \nn\\
&=& \sum_{k=\mrm{val}(v)+1}^\infty \frac{p^{-k(s-n)} - |v|^{s-n}}{|v|^{s+n}} \mrm{vol}\lb S_{n,k} \rb \nn \\ &+& \sum_{k=-\infty}^{\mrm{val}(v)-1} \frac{p^{-k(s-n)} - |v|^{s-n}}{p^{-k(s+n)}} \mrm{vol}\lb S_{n,k} \rb \nn\\
&=& \lb 1 - p^{-n} \rb \lb \sum_{k=\mrm{val}(v)+1}^\infty \frac{p^{-ks} - p^{-nk}|v|^{s-n}}{|v|^{s+n}} + \sum_{k=-\infty}^{\mrm{val}(v)-1} \lb p^{kn} - |v|^{s-n} p^{ks} \rb \rb \nn\\
&=& 0,\nn
\label{AppAeq418}
\ea
where we have used Eq. \eqref{volSni} for the volume factor. The sums converge when $\Re\lb s\rb>0$.

In the case $s-n= 2\pi i k/\ln p $ we need to evaluate the integral (when $v\neq0$)
\ba
&&J_n(v) \coloneqq \int_{\mathbb{Q}_p^n} \frac{\log_p |z-y| - \log_p |x-y|}{|z-x|^{s+n}} dz  \\ &=& \int_{\mathbb{Q}_p^n} \frac{\log_p |u| - \log_p |v|}{|u-v|^{s+n}} du \nn\\
&=& \int_{|u|<|v|} \frac{\log_p |u| - \log_p |v|}{|v|^{s+n}} du + \int_{|u|>|v|} \frac{\log_p |u| - \log_p |v|}{|u|^{s+n}} du \nn\\
&=& \sum_{k=\mrm{val}(v)+1}^\infty \frac{-k - \log_p |v|}{|v|^{s+n}} \mrm{vol}\lb S_{n,k} \rb+ \sum_{k=-\infty}^{\mrm{val}(v)-1} \frac{-k - \log_p |v|}{p^{-k(s+n)}} \mrm{vol}\lb S_{n,k} \rb \nn \\
&=& 0 \nn.
\ea
This completes the proof.
\end{proof}

\begin{thm}
When $\Re(s)>0$, $G(x,y)$ is the Green's function for the Vladimirov derivative $D_x^s$ on $\mathbb{Q}_p^n$, that is
\be
\label{appAmainresult11111}
D_x^s G\lb x,y \rb = \delta \lb x-y \rb,
\ee
acting on the Schwartz space of compactly supported locally constant functions.
\end{thm}

We will prove this theorem in several steps (lemmas).

\begin{lemma}
Let $\chi_r$ be the characteristic function of the ball
\be
B_r \coloneqq \lcb x\in \mathbb{Q}_p^n\ |\ |x| \leq p^r \rcb.
\ee
$G(x,y)$ acts as the delta function on the test function $\chi_r$.
\end{lemma}

\begin{proof}
We use direct computation of the bracket. By self-adjointness of the derivative we have
\be
\langle D^s G, \chi_r \rangle = \langle G, D^s\chi_r \rangle.
\ee
\textbf{Case 1.} Let's first consider the case $s-n \neq 2\pi i k/\ln p$. We have
\ba
&&\frac{1}{c_{n,s,p}\Gamma(\pi_{s+n})}\langle G, D^s\chi_r \rangle = \int_{\mathbb{Q}_p^n} |x-y|^{s-n} \frac{\chi_{r}\lb z \rb - \chi_{r}\lb x \rb}{|z-x|^{s+n}}dzdx \\
&=& \lb \int_{x\notin B_r} \frac{|x-y|^{s-n}}{|x|^{s+n}}dx\rb\lb \int_{z\in B_r} dz \rb\\
& & - \lb\int_{x\in B_r} |x-y|^{s-n}dx\rb \lb \int_{z\notin B_r} |z|^{-s-n} dz \rb.\nn
\label{appAeq323}
\ea
We now compute the two $z$ integrals, which converge when $n>0$, $\Re\lb s\rb>0$,
\ba
\label{appAeqq324}
\int_{z\in B_r} dz &=& \sum_{j=-\infty}^r \mrm{vol} \lb S_{n,-j} \rb = p^{nr},\\
\int_{z\notin B_r} |z|^{-s-n} dz &=& \sum_{j={r+1}}^\infty p^{-(s+n)j} \mrm{vol} \lb S_{n,-j} \rb\nn\\ 
&=&\lb 1-p^{-n} \rb \frac{p^{-sr}}{p^{s}-1}. \nn
\ea
Now we need to compute the $x$ integrals. There are two subcases: (i) $y\in B_r$ and (ii) $y\notin B_r$.

\textbf{(i)} Let's start with $y\in B_r$. We have
\ba
&&\int_{x\notin B_r} \frac{|x-y|^{s-n}}{|x|^{s+n}}dx = \int_{x\notin B_r} |x|^{-2n}dx \\&=& \sum_{j={r+1}}^\infty p^{-2nj} \mrm{vol} \lb S_{n,-j} \rb = p^{-n(r+1)}. \nn
\ea
For the $x\in B_r$ integral, note that $x'\coloneqq x-y\in B_r$, so that we can shift the integral to compute
\ba
\int_{x\in B_r} |x-y|^{s-n}dx &=& \int_{x'\in B_r} |x'|^{s-n}dx' \\
&=& \sum_{i=-\infty}^r p^{i\lb s-n \rb} \mrm{vol} \lb S_{n,-i} \rb\\
&=& \lb 1-p^{-n} \rb \frac{p^{(r+1) s}}{p^{s}-1}.
\ea
Putting everything together, we obtain, when $y\in B_r$,
\be
\frac{1}{c_{n,s,p}\Gamma(\pi_{s+n})}\langle G, D^s\chi_r \rangle = \frac{\left(1-p^{s-n}\right)
   \left(p^{-n}-p^{s}\right)}{\left(1-p^{s}\right)^2},
\ee
which implies that $\langle G, D^s\chi_r \rangle=1$, as desired.

\textbf{(ii)} Let's now consider the subcase $y\notin B_r$, and we once again compute the $x$ integrals. We have
\ba
\label{appAeq328}
\int_{x\in B_r} |x-y|^{s-n}dx &=& |y|^{s-n} \int_{x\in B_r}dx \\ &=& |y|^{s-n} (1-p^{-n}) \sum_{i=-\infty}^r p^{ni} = |y|^{s-n} p^{n r}.\nn
\ea
The last remaining integral is marginally more involved. Let $|y|\eqqcolon p^{r'}$. We have
\ba
&&\int_{x\notin B_r} \frac{|x-y|^{s-n}}{|x|^{s+n}}dx = |y|^{s-n} \int_{\substack{x\notin B_r\\|x|<|y|}} |x|^{-s-n}dx \\
& & + |y|^{-s-n} \int_{|x|=|y|} |x-y|^{s-n}dx + \int_{|x|>|y|} |x|^{-2n}dx.\nn
\ea
We now compute these three integrals, as
\ba
\label{appAeq330}
\int_{\substack{x\notin B_r\\|x|<|y|}} |x|^{-s-n}dx
&=& \sum_{i=r+1}^{r'-1} p^{-i(s+n)} \mrm{vol} \lb S_{n,-i} \rb\\
&=& \frac{\left(1-p^{-n}\right) \left(p^{-s(r+1)}-p^{-sr'}\right)}{1-p^{-s}},\nn\\
\int_{|x|>|y|} |x|^{-2n}dx &=& \sum_{i=r'+1}^\infty p^{-2ni} \mrm{vol} \lb S_{n,-i} \rb = p^{-n (r'+1)},\nn
\ea
and finally for the middle term in Eq. \eqref{appAeq330}, using Eq. \eqref{appAeq3p9} in Lemma \ref{appAlemma2}, we obtain
\be
\label{appAeq331}
|y|^{-s-n}  \int_{|x|=|y|} |x-y|^{s-n}dx = -\frac{p^{-n r'} \left(p^n+\left(p^n-2\right)
   p^{s+n}\right)}{p^{2 n}-p^{s+2n}}.
\ee
Putting together Eqs. \eqref{appAeq330}, \eqref{appAeq331}, we derive
\be
\label{appAeq332}
\int_{\substack{x\notin B_r\\|x|<|y|}} |x|^{-s-n}dx = -\frac{\left(1-p^{-n}\right) p^{-sr+r'(s-n)}}{1-p^{s}},
\ee
so that from Eqs. \eqref{appAeq323} \eqref{appAeqq324}, \eqref{appAeq328} and \eqref{appAeq332} we conclude (when $y\notin B_r$) that
\be
\langle D^s G, \chi_r \rangle = 0.
\ee 
\textbf{Case 2.} Let's now compute case $s-n = 2\pi i k/\ln p$. We have the integral
\ba\label{appAeq3233}
& &\frac{1}{c_{n,s,p}\Gamma(\pi_{s+n})}\langle G, D^s\chi_r \rangle = \lb \int_{x\notin B_r} \frac{\log_p|x-y|}{|x|^{s+n}}dx\rb\lb \int_{z\in B_r} dz \rb\\
& & - \lb\int_{x\in B_r} \log_p|x-y|dx\rb \lb \int_{z\notin B_r} |z|^{-s-n} dz \rb.\nn
\ea
When $y\in B_r$, the integrals are
\ba
&&\int_{x\notin B_r} \frac{\log_p|x-y|}{|x|^{s+n}}dx = \int_{x\notin B_r} \frac{\log_p|x|}{|x|^{s+n}}dx \\ &&= (1-p^{-n}) \sum_{i=r+1}^\infty i p^{-is}\nn\\
&=& \frac{\left(p^n-1\right) \left((r+1) p^s-r\right) p^{-n-rs}}{\left(p^s-1\right)^2} \\
&&\int_{x\in B_r} \log_p|x-y|dx = \int_{x\in B_r} \log_p|x'|dx' \\&=& (1-p^{-n}) \sum_{i=-\infty}^{r} i p^{ni}\nn \\
&=& p^{n r} \left(\frac{1}{1-p^n}+r\right).
\ea
Putting these together with Eqs. \eqref{appAeqq324} we obtain, when $y\in B_r$,
\be
\frac{1}{c_{n,s,p}\Gamma(\pi_{s+n})}\langle G, D^s\chi_r \rangle = \frac{p^{-n}+1}{p^n-1},
\ee
so that $\langle G, D^s\chi_r \rangle=1$ as desired.

When $y\notin B_r$, the integrals are 
\ba
\int_{x\in B_r} \log_p|x-y|dx &=& (1-p^{-n}) \log_p |y| \sum_{i=-\infty}^{r} p^{ni} \\&=& p^{nr} \log_p |y|, \nn \\
\int_{x\notin B_r} \frac{\log_p|x-y|}{|x|^{s+n}}dx &=& - \int_{v(x) < v(y)} \frac{v(x)}{|x|^{s+n}} dx \\&-& \int_{v(x) = v(y)} \frac{v(x-y)}{|y|^{s+n}} dx\nn\\
& & - \int_{-r>v(x)>v(y)} \frac{v(y)}{|x|^{s+n}} dx. \nn
\ea
Denote $i \coloneqq v(x)$ and $b \coloneqq v(y)$. The first integral is 
\ba
\int_{v(x) < v(y)} \frac{v(x)}{|x|^{s+n}} dx &=& (1-p^{-n}) \sum_{i=-\infty}^{b-1} ip^{is}\\
&=& \frac{\left(1-p^{-n}\right) p^{b s} \left((b-1)p^s-b\right)}{\left(p^s-1\right)^2}.
\ea
To evaluate the second integral, we let $u = x-y$ and $j = v(u)$. We have this integral equal to
\ba
&&\int_{v(x) = v(y)} \frac{v(x-y)}{|y|^{s+n}} dx =\int_{j > b} \frac{j}{|u+y|^{s+n}} du\\ &+& \int_{j = b} \frac{b}{p^{-(s+n)b}}\nn\\
&=& bp^{sb}(1 - 2p^{-n})+(1-p^{-n})\sum_{j>b} j p^{(s+n)b - nj}\\
&=& bp^{sb}(1-p^{-n}) + \frac{p^{sb-n}}{1-p^{-n}}. \nn
\ea
For the third integral, we have
\ba
\int_{-r>v(x)>v(y)} \frac{v(y)}{|x|^{s+n}} dx &=& \sum_{-r>i>b}bp^{is}(1-p^{-n})\\
&=& b(1 -p^{-n})\frac{p^{s(-r-1)}-p^{bs}}{1 - p^{-s}}.
\ea
Adding the three contributions, the result of this integral, using that $s = n + 2\pi ik/\ln p$, is simply
\be 
\int_{x\notin B_r} \frac{\log_p|x-y|}{|x|^{s+n}}dx = -b p^{n (-r-1)}.
\ee
Plugging this back into \eqref{appAeq3233}, we obtain that it equals zero. This completes the proof.
\end{proof}

Below is the technical result we have used.

\begin{lemma}
\label{appAlemma2}
 For $y\in \mathbb{Q}_p^n$, $|y|\eqqcolon p^{r'}$ and $\Re\lb t\rb> -n $, we have the integral
\be
\int_{|x|=|y|} |x-y|^{t} dx = \frac{\left(\left(p^n-2\right) p^{n+t}+1\right) p^{n(r'-1)+r' t}}{p^{n+t}-1}.
\ee
\label{appAeq3p9}
\end{lemma}
\begin{proof}
	We proceed by direct computation. We write $x$ as $x = u + y$ and divide this integral into two parts according to the $\mathbb{Q}_p^n$ valuation $v(u)$,
	\be \int_{|x|=|y|} |x-y|^{t} dx = \int_{v(u) > v(y)} |x-y|^t dx  + \int_{\substack{v(u) = v(y)\\ v(x) = v(y)}} |x-y|^t dx. \ee
	We evaluate the first integral, denoting $v_n(u) = i$ and using the volume factor derived in Lemma \ref{lemma1},
	\ba 
	\int_{v(u) > v(y)} |x-y|^t dx  &=& \int_{v(u) > v(y)} |u|^t du \\
	&=& (1 - p^{-n})\sum_{i \geq v(y)+1} p^{-(n+t)i} \\
	&=& (1 -p^{-n})\frac{p^{-(n+t)(v(y)+1)}}{1 - p^{-(n+t)}}.
	\ea
	The second integral is 
	\be
	\int_{\substack{v(u) = v(y) \\ v(x) = v(y)}} p^{-v(y)t} du = p^{-(n+t)v(y)}(1 - 2p^{-n}). 
	\ee
	Here measure factor $p^{-nv(y)}(1 - 2p^{-n})$ obtained by subtracting the neighborhood of $u = -y$, since in that case $v(x) \neq v(y)$. 
	
	Summing up the two contributions, we obtain
	\ba 
	\int_{|x|=|y|} |x-y|^{t} dx &=& \frac{(p^{n+t}(p^n -2) + 1)p^{n(r' -1) + r't}}{p^{n+t} -1},
	\ea
	as advertised.
\end{proof}

\end{document}